\documentclass[english]{llncs}

\usepackage{amssymb,amsmath}
\usepackage{graphicx}

\usepackage{multirow}

\usepackage{algorithmicx}
\usepackage[noend]{algpseudocode}
\usepackage{algorithm}

\usepackage{enumitem}

\usepackage{xspace}

\usepackage{hyperref}

\usepackage{ubpalgo}

\let\Oldendproof\endproof%
\def\endproof{\qed\Oldendproof}%

\let\Oldendremark\endremark%
\def\endremark{\hfill\ensuremath{\diamondsuit}\Oldendremark}%

\let\Oldendexample\endexample%
\def\endexample{\hfill\ensuremath{\triangle}\Oldendexample}%

\spnewtheorem{sublemma}{Claim}[sublemmacounter]{\itshape}{\rmfamily}
\numberwithin{sublemma}{lemma}

\spnewtheorem*{subproof}{Proof}{\itshape}{\rmfamily}
\let\Oldendsubproof\endsubproof
\def\endsubproof{\hfill\ensuremath{\Diamond}\Oldendsubproof}%

\newcommand{\locations}{\ensuremath{V}}
\newcommand{\aloc}{\ensuremath{loc}}
\newcommand{\acnum}{\ensuremath{\Delta x}}

\newcommand{\loadd}{\ensuremath{x_m}}
\newcommand{\driver}{\ensuremath{driv}}
\newcommand{\ndriver}{\ensuremath{k}}
\newcommand{\orig}{\ensuremath{orig}}
\newcommand{\origin}{\ensuremath{orig}}
\newcommand{\dest}{\ensuremath{dest}}
\newcommand{\tdep}{\ensuremath{dep}}
\newcommand{\tarr}{\ensuremath{arr}}
\newcommand{\mloadd}{\ensuremath{\ell oad}}

\newcommand{\ntourd}{\ensuremath{n}}

\newcommand{\taskset}{\ensuremath{\mathcal{T}}}

\newcommand{\abs}[1]{\ensuremath{\lvert #1 \rvert}}

\newcommand{\NN}{\ensuremath{\mathbb{N}}}
\newcommand{\ZZ}{\ensuremath{\mathbb{Z}}}

\newcommand{\RR}{\ensuremath{\mathbb{R}}}

\newcommand{\capd}{\ensuremath{L}}
\newcommand{\tourd}{\ensuremath{\Gamma}}
\newcommand{\action}{\ensuremath{\alpha}}

\newcommand{\aexetime}{\ensuremath{t}}
\newcommand{\move}{\ensuremath{m}}
\newcommand{\dist}{\ensuremath{d}}

\newcommand{\tact}{\ensuremath{act}}
\newcommand{\tmov}{\ensuremath{mov}}

\newcommand{\costd}{\ensuremath{\text{cost}^{\text{d}}}}
\newcommand{\costc}{\ensuremath{\text{cost}^{\text{c}}}}

\newcommand{\fc}{\ensuremath{f}}
\newcommand{\fd}{\ensuremath{F}}

\newcommand{\bd}{\ensuremath{B}}
\newcommand{\T}{\ensuremath{\mathcal{T}}}
\newcommand{\capacity}{\ensuremath{\text{cap}}}

\newcommand{\sched}{\mathcal{S}}
\newcommand{\schedule}{\sched}
\newcommand{\TR}{\mathrm{TR}}

\newcommand{\x}{\boldsymbol{x}}
\newcommand{\z}{\boldsymbol{z}}
\newcommand{\zd}{\boldsymbol{Z}^D} % the vector
\newcommand{\zzd}{Z^D} % an entry of the vector

\newcommand{\REOPT}{\textsc{ReOpt}\xspace}
\newcommand{\PDPINSERT}{\textsc{Pdp-Insert}\xspace}

\newcommand{\VO}{V_O}
\newcommand{\VU}{V_U}

\newcommand{\VD}{V_D}

\newcommand{\Vpick}{{V^+}}
\newcommand{\Vdrop}{{V^-}}
\newcommand{\Vbal}{{V^=}}

\newcommand{\NPhard}{\ensuremath{\mathcal{NP}}\hbox{-}hard}

\allowdisplaybreaks[2]

\title{ReOpt: an Algorithm with a Quality Guaranty for Solving the Static Relocation Problem}

\author{Sahar Bsaybes\inst{2}${}^{\star}$
        \and
        Sven~O.~Krumke\inst{1}
        \and
        Alain Quilliot\inst{2}
        \and
        Annegret K.\ Wagler\inst{2}
        \and
        Jan-Thierry Wegener\inst{2}\thanks{This work was founded by the French National Research Agency, the European Commission (Feder funds) and the R\'egion Auvergne within the LabEx IMobS3.}}
\institute{University of Kaiserslautern (Department of Mathematics) \\
           Kaiserslautern, Germany\\
           \texttt{krumke@mathematik.uni-kl.de}
          \and
          Universit\'e Blaise Pascal (Clermont-Ferrand II)\\
          Laboratoire d'Informatique, de Mod\'elisation et d'Optimisation des Syst\`emes\\
          BP 10125, 63173 Aubi\`ere Cedex, France\\
        \texttt{\{bsaybes,quilliot,wagler,wegener\}@isima.fr}}

\begin{document}

\maketitle

%%%%%%%%%%%%%%%%%%%%%%%%%%%%%%%%%%%%%%%%%%%%%%%%%%%%%%%%%
\begin{abstract}
In a carsharing system, a fleet of cars is distributed at stations in an urban area, customers can take and return cars at any time and station.
For operating such a system in a satisfactory way, the stations have to keep a good ratio between the numbers of free places and cars in each station.
This leads to the problem of relocating cars between stations, which can be modeled within the framework of a metric task system. 
In this paper, we focus on the Static Relocation Problem, where the system has to be set into a certain state, outgoing from the current state.  
We present a combinatorial approach and provide approximation factors for several different situations.
\end{abstract}
%%%%%%%%%%%%%%%%%%%%%%%%%%%%%%%%%%%%%%%%%%%%%%%%%%%%%%%%%

%%%%%%%%%%%%%%%%%%%%%%%%%%%%%%%%%%%%%%%%%%%%%%%%%%%%%%%%%%%%%%%%%%%%%%%%%%%%%%%%%%%%%%%%%%%%%%%%%%%%%%%%%%

\section{Introduction}

Carsharing is a modern way of car rental, attractive to customers who make only occasional use of a car on demand.
Carsharing contributes to sustainable transport as less car intensive means of urban transport, and an increasing number of cities all over the world establish(ed) such services.
Hereby, a fleet of cars is distributed at specified stations in an urban area, customers can take a car at any time and station and return it at any time and station,
provided that there is a car available at the start station and a free place at the final station.
To ensure the latter, the stations have to keep a good ratio between the number of places and the number of cars in each station.
This leads to the problem of balancing the load of the stations, called \emph{Relocation Problem}:
an operator has to monitor the load situations of the stations and to decide when and how to move cars from ``overfull'' stations to ``underfull'' ones.

Balancing problems of this type occur for any car- or bikesharing system, but the scale of the instances, 
the time delay for prebookings and the possibility to move one or more vehicles in balancing steps differ. 
We consider an innovative carsharing system, where the cars are partly autonomous, which allows to build wireless convoys of cars leaded by a special vehicle,
such that the whole convoy is moved by only one driver (cf.~\cite{EDGC:2012:PCS}).

This setting is different from usual carsharing, but similar to bikesharing, where trucks can simultaneously load and move several bikes during the relocation process \cite{Benchimol+etal:RAIRO,do-cmc2013,cirrelt-CMR-2012}.
The main goal is to guarantee a balanced system during working hours (dynamic situation as in \cite{cirrelt-CMR-2012,HKQWW:2015:ODY,HKQWW:LAGOS:2015,KQWW:2014:LNCS,LAGOS2013})
or to set up an appropriate initial state for the morning (static situation as in \cite{do-cmc2013,LAGOS2013}).
Both, the dynamic and the static versions are known to be \NPhard~\cite{Ball+etal:handbook:95a,Ball+etal:handbook:95b,do-cmc2013,Nemhauser+etal:handbook:89}, and different heuristic approaches have been developed, e.g., 
partitioning the problem into subproblems by
discrete particle swarm optimization \cite{CIE:GKA-2013},
by reducing the search space~\cite{KQWW:2014:LNCS},
using variable neighborhood search \cite{EvoCOP:HPHR-2013}, or
partitioning the problem into subproblems with clustering techniques \cite{SHH-2013}.

%%%%%%%%%%%%%%%%%%%%%%%%%%%%%%%%%%%%%%%%%%%%%%%%%%%%%%%%%%%%%%%%%%%%%%%%%%%%%%%%%%%%%%%%%%%%%%%%%%%%%%%%%%

In this paper, we address the static situation, where the system has to be set into a certain state~$\z^T$, outgoing from the current state~$\z^0$, within a given time horizon.~$[0, T]$ 

In Section \ref{sec: model}, we model different variants of the Static Relocation Problem within the framework of a (quasi) metric task system,
where tasks consist in moving cars out of ``overfull'' stations (with $z^0_v > z^T_v$) into ``underfull'' stations (with $z^0_v < z^T_v$).

In the following, we present two approaches to solve the Static Relocation Problem, an exact integer programming based method and a combinatorial approximation algorithm.

In Section~\ref{sec: static: min-cost flows: wo pre w back}, we present an exact approach to solve this problem using coupled flows in time-expanded networks,
where the flow of cars is dependent from the flow of drivers since cars can only be moved in convoys.
Due to the involved flow coupling constraint, all variants of the problem require very long computation times (see also Section~\ref{seq: computational})
which motivated the study of a combinatorial algorithm to solve the problem in reasonably short time.

In Section~\ref{sec: static: reopt}, we propse such an algorithm that firstly matches tasks to generate transport requests,
subsequently solves a Pickup and Delivery Problem, and iteratively augments the transport requests and resulting tours (Section~\ref{sec: static: reopt}). 
This algorithm has an approximation factor based on the convoy size in several different situations (Section~\ref{sec: static: reopt: approximation factor}).
Parts of these results appeared without proofs in~\cite{LAGOS2013}.
 
We finally provide some computational results for both approaches (Section~\ref{seq: computational}), and close with some remarks and future lines of research (Section~\ref{seq: conclusion}).

%%%%%%%%%%%%%%%%%%%%%%%%%%%%%%%%%%%%%%%%%%%%%%%%%%%%%%%%%%%%%%%%%%%%%%%%%%%%%%%%%%%%%%%%%%%%%%%%%%%%%%%%%%

\section{Problem description and model}
\label{sec: model}

We model the Relocation Problem in the framework of a metric task system. 

By \cite{LAGOS2013}, the studied carsharing system can be understood as a discrete event-based system, where
the system components are the stations $v_1, \ldots, v_n$, each having an individual capacity cap$(v_i)$,
a system state $\z^t \in \ZZ^n$ specifies for each station $v_i$ the number of cars $z^t_i$ at a time point $t \leq T$ within a time horizon~$[0, T]$ 
and $\z^t$ changes when customers or convoy drivers take or return cars at stations.

An operator monitors the evolution of system states over time and decides when and how to move cars
between the stations,
in order to avoid infeasible system states $\z^t$ with $z^t_v > \capacity(v_i)$ or $z^t_v < 0$ for some station $v$.
More precisely, a \emph{task} is defined by $\tau = (v, x)$ where $x \in \ZZ \setminus \{0\}$ cars are to pickup (if $x > 0$) or to deliver (if $x < 0$) at station $v$ within the time-horizon~$[0, T]$.

In the static situation, we consider a \emph{start state}~$\z^0$ and a \emph{destination state}~$\z^T$.
Then the tasks are induced by the differences between these two states, i.e, for every station $v \in V$ with $z^0_v \neq z^T_v$ we have a task~$(v, z^0_v - z^T_v)$.
A station $v_o$ with $z^0_{v_o} - z^T_{v_o} > 0$ is called an~\emph{overfull station}, a station~$v_u$ with $z^0_{v_u} - z^T_{v_u} > 0$ is called an~\emph{underfull station},
and a station $v_b$ with $z^0_{v_o} - z^T_{v_u} = 0$ is called an~\emph{balanced station}.

Since our considered carsharing system is located in an urban area, due to one-way streets, the distances between two stations~$v$ and~$v'$ may not be symmetric, i.e.,
the distance from~$v$ to~$v'$ 
may differ from that from~$v'$ to~$v$.
Furthermore, due to the traffic situation, the travel times also may vary.
In order to represent the one-way streets (or different travel times) it is suitable to encode the urban area where the carsharing system is running as a
\emph{(quasi) metric space}\footnote{A quasi metric space $(M,d)$ contains a set $M$ and a distance function $d: M \times M \to \RR$ so that for all $x,y,z \in M$ holds:
$(i)$ $d(x,y) \geq 0$ (non-negativity), $(ii)$ $d(x, y) = 0$ iff $x = y$ (identity of indiscernibles), $(iii)$ $d(x,z) \leq d(x,y) + d(y,z)$ (triangle inequality).
Laxly said, a quasi metric is a ``metric'' where the symmetry condition is dropped.}
$M=(V,d)$ induced by a directed weighted graph $G = (V,A,\capacity,w)$, where the nodes correspond to stations, arcs to their physical links in the urban area, node weights to the station's capacities,
and the edge weights $w : E \rightarrow \RR_+$ determine the distance $d$ between two points $v,v' \in V$ as length of a shortest directed path from~$v$ to~$v'$.

All drivers begin and end their work at the same location, a so-called depot.
A depot is represented in $V$ by a distinguished origin $v_D \in V$.
In larger carsharing systems, there are usually several depots~$V_D$ distributed within the urban area.
Especially when a driver lives near a depot, he may prefer to start and end its tour in the same depot.
Furthermore, this ensures that always the same number of trucks are located in the depots.
However, in order to decrease the total distances traveled by the drivers, it may be suitable for the operators of the carsharing system that the drivers start and end their tours in different depots.
The number of drivers starting in depot~$v_D \in \VD$ is represented by a natural number~$\zzd_{v_D}$; the total number of available drivers is denoted by~$k = \abs{\zd}$.

This together yields a \emph{(quasi) metric task system}, a pair $(M, {\cal T})$ where $M = (V, d)$ is the above metric space and ${\cal T}$ a set of tasks, as suitable framework to embed the tours for the convoys. 
A driver able to lead a convoy plays the role of a server.
Each server has capacity~$L$, corresponding to the maximum possible number of cars per convoy; several servers are necessary to serve a task $\tau = (v, x)$ if $\abs{x} > L$ holds.

More precisely, we define the following.
An \emph{action} for driver~$j$ is a $4$-tuple $\action = (j, v, t, x)$, where 
$j \in \{ 1, \dotsc, \ndriver \}$ specifies the driver $\driver(\action)$ performing the action, 
$v \in \locations$ specifies the station $\aloc(\action)$, 
$t \in [0, T ]$ is the time $t(a)$ when the action is performed, 
and $x \in \ZZ$ the number of cars $\acnum(\action)$ to be loaded (if $x > 0$) or unloaded (if $x < 0$). 
Hereby, the capacity of the convoy must not be exceeded, i.e., we have $\abs{x} \leq \capd$.
We say that an action is \emph{performed} (by a driver) if he loads (resp. unloads) $\abs{x}$ cars at~$v$.
An action is called \emph{pickup action} if $x > 0$, \emph{drop action} if~$x < 0$, and \emph{empty} if $x = 0$.

A \emph{move} from one station to another is
$\move = (j, v, t_v, v', t_{v'}, \loadd)$, where
$j \in \{ 1, \dotsc, \ndriver \}$ specifies the driver $\driver(\move)$ that has to move from the origin station $\origin(\move) = v \in \locations$ starting at time $\tdep(\move) = t_v$
to destination station $\dest(\move) = v' \in \locations$ arriving at time $\tarr(\move) = t_{v'}$, and a load of $\mloadd(\move) = \loadd$ cars.
Hereby, we require that 
\begin{enumerate}
 \item\label{def: enum: move: 0} $0 \leq \mloadd(\move) \leq \capd$,
%  \item\label{def: enum: move: 1} $\pathd$ is a shortest $(v,w)$-path,
 %\item\label{def: enum: move: 1} \jan{the move is on a shortest path between the stations, i.e., $t_{v'} - t_v = d(v, v')$, and}
 \item\label{def: enum: move: 1} $t_{v'} - t_v = d(v, v')$, and
 \item\label{def: enum: move: 2} from $\orig(\move) \neq \dest(\move)$ follows $\tarr(\move) = \tdep(\move) + \dist(\orig(\move), \dest(\move))$.
\end{enumerate}
A move $(j, v, t^v, w, t^w, \loadd)$ with $v = w$ is called \emph{waiting move}.

A \emph{tour} is an alternating sequence $\tourd = (\move^1, \action^1, \move^2, \action^2, \dotsc, \action^{\ntourd - 1}, \move^\ntourd)$ of moves and actions with
\begin{enumerate}
 \item\label{def: enum: tour: 1} $
\driver(\move^1) = \driver(\action^1) = \dotsm = \driver(\action^{\ntourd - 1}) = \driver(\move^\ntourd)$, 
 \item\label{def: enum: tour: 2} $\dest(\move^i) = \aloc(\action^i) = \orig(\move^{i+1})$,
 \item\label{def: enum: tour: 3} $\tarr(\move^i) = t(\action^i) = \tdep(\move^{i+1})$, and
 \item\label{def: enum: tour: 5} $\mloadd(\move^{i+1}) = \mloadd(\move^{i}) + \acnum(\action^i)$.
\end{enumerate}
By $\tact(\tourd)$ we denote the sequence containing the actions of the tour $\tourd$, i.e., $\tact(\tourd) = (\action^1, \dotsc, \action^{\ntourd-1})$,
and by $\tmov(\tourd)$ we denote the sequence containing the moves of the tour $\tourd$, i.e., $\tmov(\tourd) = (\move^1, \dotsc, \move^\ntourd)$.

The \emph{length of a tour} corresponds to the distance traveled by the driver.
Several tours are composed to a transportation schedule. 
A collection of tours $\{\tourd_1, \ldots, \tourd_\ndriver \}$ is a \emph{feasible transportation schedule} $\sched$ for $(M, \taskset)$ if 
\begin{enumerate}
 \item\label{def: enum: schedule: 1} every driver has exactly one tour,
 \item\label{def: enum: schedule: 2} each task $\tau \in \taskset$ is served (i.e., for every task $\tau = (v, x)$, the number of cars picked up (resp.~dropped) at station~$v$ sum up to~$x$),
 \item\label{def: enum: schedule: 3} all system states $\z^t$ are feasible during the whole time horizon $[0, T]$. 
\end{enumerate}
The \emph{total tour length} of a transportation schedule is the sum of the lengths of its tours.
Condition~\ref{def: enum: schedule: 3} requires that, besides the canonical precedences between a move $\move_i \in \tourd$ and its successor move $\move_{i+1}$ in the same tour~$\tourd$,
also dependencies between tours are respected
if \emph{preemption} is used, i.e., if a car is transported in one convoy from its origin to an intermediate station, and from there by another convoy to its destination.
This causes dependencies between tours, since some moves cannot be performed before others are done without leading to infeasible intermediate states (the reason why tours may contain waiting moves).
More precisely, an action $\action \in \tourd$ induces a precedence, avoiding a system state $z^t$ with $z^t_v < 0$ (resp.~$cap(v) < z^t_v$),
if one of the following conditions is true:
\begin{itemize}
\item $\action$ drops cars at an overfull station,
\item $\action$ picks up cars at an underfull station,
\item $\action$ drops or picks up cars at a balanced station.
\end{itemize}
Note, it is possible that the action which is dependent on $\action$ may not be uniquely determinable.
Furthermore, an action with a dependency to $\action$ does not necessarily fulfill any of the three conditions.

We call a transportation schedule \emph{non-preemptive} if there are no precedences between moves of different tours, and \emph{preemptive} otherwise.
When every driver starts and ends its tour in the same depot, we say it is a transportation schedule \emph{with backhaul}.
Otherwise, we call it a transportation schedule \emph{without backhaul}.

Our goal is to construct non-preemptive transportation schedules of minimal total tour length for the Relocation Problem in the static situation.

\begin{problem}[Static Relocation Problem $(M,\z^0,\z^T,\zd,k,L)$]
Given a (quasi) metric space $M=(V, \dist)$ induced by a (directed) weighted graph $G=(V, E, w, \capacity)$,
start state $\z^0 \in \NN^{|V|}$, destination state $\z^T \in \NN^{|V|}$ with $\abs{\z^0} = \abs{\z^T}$ and time horizon $[0,T]$, $k$ servers of capacity $L$ and $\zzd_{v_d}$ drivers start at depot $v_d \in V_D$,
find a non-preemptive transportation schedule of minimal total tour length for the metric task system $(M, \T)$ where $\T$ consists of the tasks $\tau = (v_i, z_i^0-z_i^T)$ for all~$v_i$ with~$z_i^0 \neq z_i^T$.

Hereby, we further classify the Static Relocation Problem into the following problems and their compositions:
\begin{itemize}
 \item Symmetric Static Relocation Problem: if $M$ is a metric space,
 \item Asymmetric Static Relocation Problem: if $M$ is a quasi metric space,
 \item Static Relocation Problem with backhaul: if the solution is a transportation schedule with backhaul, and
 \item Static Relocation Problem without backhaul: if the solution is a transportation schedule without backhaul.
\end{itemize}
\end{problem}

%%%%% %%%%% %%%%%
\section{Min-Cost Flows in Time-Expanded Network}
\label{sec: static: min-cost flows: wo pre w back}

In this section, we give an exact approach for the Static Relocation Problem $(M,\z^0,\z^T, \zd, k,L)$ without preemption by defining a time-expanded network with coupled flows: car and driver flows.
Hereby, we firstly describe the approach for the Static Relocation Problem with backhaul (Sections~\ref{sec: static: min-cost flows: ten: wo pre w back}--\ref{sec: static: min-cost flows: ilp: wo pre w back}).
Afterwards, in Section~\ref{sec: static: min-cost flows: without backhaul}, we show how this approach can be modified in order to solve the Static Relocation Problem without backhaul.

We consider a (quasi) metric space $M = (V, \dist)$ induced by a (directed) weighted graph $G=(V \cup \VD, E, w, \capacity)$ representing the set of stations~$V$, a set of depots $\VD$,
the road (or logical) connections $E$ between them, driving times $w \colon E \to \NN$, and the (quasi) metric $\dist$ induced by the shortest path distances in $G$. 
In addition, there are per unit costs $\costc$ and $\costd$ for moving cars and drivers within $G$. 
The task set ${\T}$ consists of the tasks $\tau = (v, z_v^0-z_v^T)$ for all $v$ with $z_v^0 \neq z_v^T$. 
The output is a preemptive transportation schedule for the metric task system $(M, {\T})$, the objective is to minimize its total tour length.

For that, we build a directed graph~$G_T = (V_T, A_T)$, with $A_T = A_H \cup A_L$, as a time-expanded version of the original network~$G$. 
For each $1 \leq j \leq k$, the drivers and their convoys will form flows~$\fd_j$ and $\fc_j$ through~$G_T$ which are coupled in the sense that on those arcs $a \in A_L$ used for moves of convoys,
we have the condition $\fc_j(a)\leq \capd \cdot \fd_j(a)$ reflecting the dependencies between the two flows.

%%% %%% %%%
\subsection{Time-expanded network \texorpdfstring{$G_T$}{GT}}
\label{sec: static: min-cost flows: ten: wo pre w back}

We build a time-expanded version~$G_T = (V_T, A_T)$ of the original network~$G$.

The node set $V_T$ is constructed as follows: 
for each station and each depot~$v \in V$ and each time point~$t$ in the given time horizon $[0, T]$, there is a node $(v,t) \in V_T$ which represents station/depot~$v$ at time~$t$.

The arc set $A_T = A_H \cup A_L$ of~$G_T$ is composed of two subsets:
\begin{itemize}
\item $A_H$ contains, for each station~$v \in V$ of the original network and each $t \in \{ 0, 1, \dotsc, T-1 \} $, the \emph{holdover arc} connecting~$(v,t)$ to~$(v,t+1)$. 
\item $A_L$ contains, for each edge $(v,v')$ of~$G$ and each point in time~$t \in \{ 0, \dotsc, T \}$ such that $t + \dist(v,v') \leq T$, the \emph{relocation arc} from $(v,t)$ to $(v',t + \dist(v,v'))$.
\end{itemize}
Note, due to the identity of indiscernibles (i.e., $d(v,v') = 0$ iff $v = v'$) it follows that the time-expanded network $G_T$ is acyclic by construction.
Furthermore, $G_T$ is constructed the same way, regardless whether the original network is directed or undirected.

%%% %%% %%%
\subsection{Flows in \texorpdfstring{$G_T$}{GT}}
\label{sec: static: min-cost flows: flows: wo pre w back}

On the relocation arcs of~$G_T$, we define for each driver~$i$ a driver flow $\fd_i$ as well as a car flow $\fc_i$ which represents the convoy of this driver.
On the holdover arcs of~$G_T$, we define a driver flow $\fd_i$ for each driver but only one car flow $\fc$.
We specify the capacities as well as the costs for each arc with respect to the different flows.

A flow on a relocation arc corresponds to a (sub)move in a tour, i.e., some cars are moved by drivers in a convoy from station~$v$ to another station~$v'$.
Hereby, the stations can be used to pick up or to drop cars, or simply to transit a node (when a driver/convoy passes the station(s) on its way to another station).
A relocation arc from $(v,t)$ to $(v', t + \dist(v,v'))$ has capacity~$1$ for each of the driver flow $\fd_i$, $1 \leq i \leq k$.
In order to ensure that cars are moved only by drivers and only in convoys of capacity~$\capd$, we require that 
\begin{align} \label{eq: static: min-cost flows: flows: wo pre w back: 1}
  \fc_i(a) \leq \capd \cdot \fd_i(a) \text{ for all $a \in A_L$.}
\end{align}
Thus, the capacities for the car flows $\fc_i$ on the relocation arcs are not given by constants but by a function.
Note that due to these flow coupling constraints, the constraint matrix 
of the network is not totally unimodular (as in the case of uncoupled flows) and therefore solving such problems becomes hard.
From Equation~\eqref{eq: static: min-cost flows: flows: wo pre w back: 1} it directly follows that $\fc_i(a) \leq \capd$ holds for all relocation arcs $a \in A_L$.

Since we consider transportation schedules without preemption, we must ensure that cars are not exchanged between convoys.
Furthermore, we have to ensure that there is no ``tour internal preemption'', i.e., a tour drops a car at one station and later picks it up again.
Inner preemption and preemption between tours can be avoided by ensuring that cars are only picked up at overfull stations and dropped at underfull stations.
This means that there are more cars in a convoy leaving an overfull station than entering the station.
Analogously, there are more cars in a convoy entering an underfull station than leaving the station, and in balanced stations the number of cars entering the station is equal to the number of cars leaving the station.
Thus, one has to ensure that the following constraints hold for every $1 \leq i \leq k$
\begin{align}
  & \sum_{a \in \delta^+_L(v_o, t)} \fc_i(a) \leq \sum_{a \in \delta^-_L(v_o, t)} \fc_i(a), && 
      \begin{aligned} 
          & \text{ for all overfull stations } v_o \in V_O \\
          & \text{ and all $0 < t < T$},
      \end{aligned} \label{eq: static: min-cost flows: flows: wo pre wo back: over under balanced: 1} \\
  & \sum_{a \in \delta^+_L(v_u, t)} \fc_i(a) \geq \sum_{a \in \delta^-_L(v_u, t)} \fc_i(a), &&
      \begin{aligned}
          & \text{ for all underfull stations } v_u \in V_U \\
          & \text{ and all $0 < t < T$},
      \end{aligned} \label{eq: static: min-cost flows: flows: wo pre wo back: over under balanced: 2} \\
  & \sum_{a \in \delta^+_L(v_b, t)} \fc_i(a) = \sum_{a \in \delta^-_L(v_b, t)} \fc_i(a), &&
      \begin{aligned}
          & \text{ for all balanced stations } v_b \in V_B \\
          & \text{ and all $0 < t < T$},
      \end{aligned} \label{eq: static: min-cost flows: flows: wo pre wo back: over under balanced: 3} 
\end{align}
where $\delta^+_L(v,t)$ denotes the set of incoming relocation arcs of $(v,t)$, and $\delta^-_L(v,t)$ denotes the set of outgoing relocation arcs of $(v,t)$.
Note, that the sums in these equations are of a technical nature.
In fact, in each of these equations there is at most one summand positive while all other are zero.
However, since we do not know in advance which one is positive, we have to sum over all incoming resp.~outgoing relocation arcs of a node.

The costs for the driver flows on a relocation arc~$a$ correspond to the distance traveled, i.e., if $a$ corresponds to the edge~$(v,v') \in E$ the costs correspond to $w(a) := w(v, v')$.

%%%%%%%%%%%%%%%%
A flow on a holdover arc corresponds to cars/drivers remaining at the station in the time interval~$[t,t+1]$.
Since we do not allow precedences between tours, a car is dropped only at underfull stations and only picked up from overfull stations;
at balanced stations, cars are neither picked up nor dropped (see constraints~\eqref{eq: static: min-cost flows: flows: wo pre wo back: over under balanced: 1}--\eqref{eq: static: min-cost flows: flows: wo pre wo back: over under balanced: 3}).
Therefore, we define only one car flow~$\fc$ on the holdover arcs.
Furthermore, it follows from~\eqref{eq: static: min-cost flows: flows: wo pre wo back: over under balanced: 1}--\eqref{eq: static: min-cost flows: flows: wo pre wo back: over under balanced: 3}
that at an overfull station the number of cars is non-increasing over time, at an underfull station the number of cars is non-decreasing over time, and at balanced stations the number of cars remains equal all the time.
Thus, it holds for every reachable system state~$\z$
\begin{align*}
  & z^0_v \geq z_v \geq z^T_v, & \text{ if $v$ is an overfull station,} \\
  & z^0_v \leq z_v \leq z^T_v, & \text{ if $v$ is an underfull station,} \\
  & z^0_v = z_v = z^T_v, & \text{ if $v$ is a balanced station.}
\end{align*}
Since $\z^0$ and $\z^T$ are feasible system states by definition, this implies that every reachable system state~$\z$ is feasible as well.
Thus, capacities are not needed for holdover arcs with respect to the car flow~$\fc$.
For each driver flow~$\fd_i$, $1 \leq i \leq k$, we set a capacity of~1 on the holdover arcs.
Moreover, the cost for all flows on such arcs are zero.

%%%%%%%%%%%

With the help of~$\zd$, we can assign a depot to each driver.
For that, we set $\bd_{v_D}^i = 1$ if driver~$i$ starts (and ends) its tour in the depot~$v_D \in \VD$, and $\bd_{v_D}^i = 0$ otherwise.

To correctly initialize the system, we use the nodes $(v, 0) \in V_T$ as sources for both flows and set their 
balances accordingly to the initial numbers of cars at station~$v$ and time~$0$ and locate the drivers at their depot $v_D \in \VD$, i.e.,
the sum of the car flow values on all outgoing arcs of $(v, 0)$ is set to $z_v^0$, and for each driver~$i$, $1 \leq i \leq k$, the sum of the driver flow values on all outgoing arcs of $(v_D, 0)$ is set to~$\bd_{v_D}^i$.

For all internal nodes $(v, t) \in V_T$ with $0 < t < T$, we use flow conservation constraints for the car flows~$\fc$ and $\fc_i$, i.e.,
\[
    \begin{aligned}
    \fc((v, t-1), (v, t)) + \sum_{a\in \delta^+_L(v,t)} \sum_{i=1}^k \fc_i(a) 
        = \fc((v, t), (v, t+1)) + \sum_{a \in \delta^-_L(v,t)} \sum_{i=1}^k \fc_i(a),
    \end{aligned}
\]
and the standard flow conservation constraints for each driver flows~$\fd_i$, i.e.,
\[
  \sum_{a\in \delta^-(v,t)} \fd_i(a) = \sum_{a\in \delta^+(v,t)} \fd_i(a),
\]
for all~$1 \leq i \leq k$.

To ensure that the destination state is reached and each driver returns to its depot, we use as sinks the nodes $(v, T)$, $v \in V$, for the car flow 
and the nodes $(v_D, T)$, $v_D \in \VD$, for the driver flow, and set their balances accordingly to $\z^T$ resp.~to~$\bd_{v_D}^i$ for each driver flow~$\fd_i$.
Since there are $k$~drivers, the sum of all outgoing (resp.~incoming) driver flows of all depots $v_D \in \VD$ sums up to~$k$.

Figure \ref{fig: flownetwork} illustrates a time-expanded network with capacities on the arcs as well as the balances for the nodes~$(v, 0)$ and $(v, T)$.
Furthermore, optimal flows $\fc, \fc_i$ and $\fd_i$ based on the network from Figure \ref{fig: reopt} are shown in this figure.

\begin{figure}
  \centering
  \includegraphics[width=0.962\textwidth]{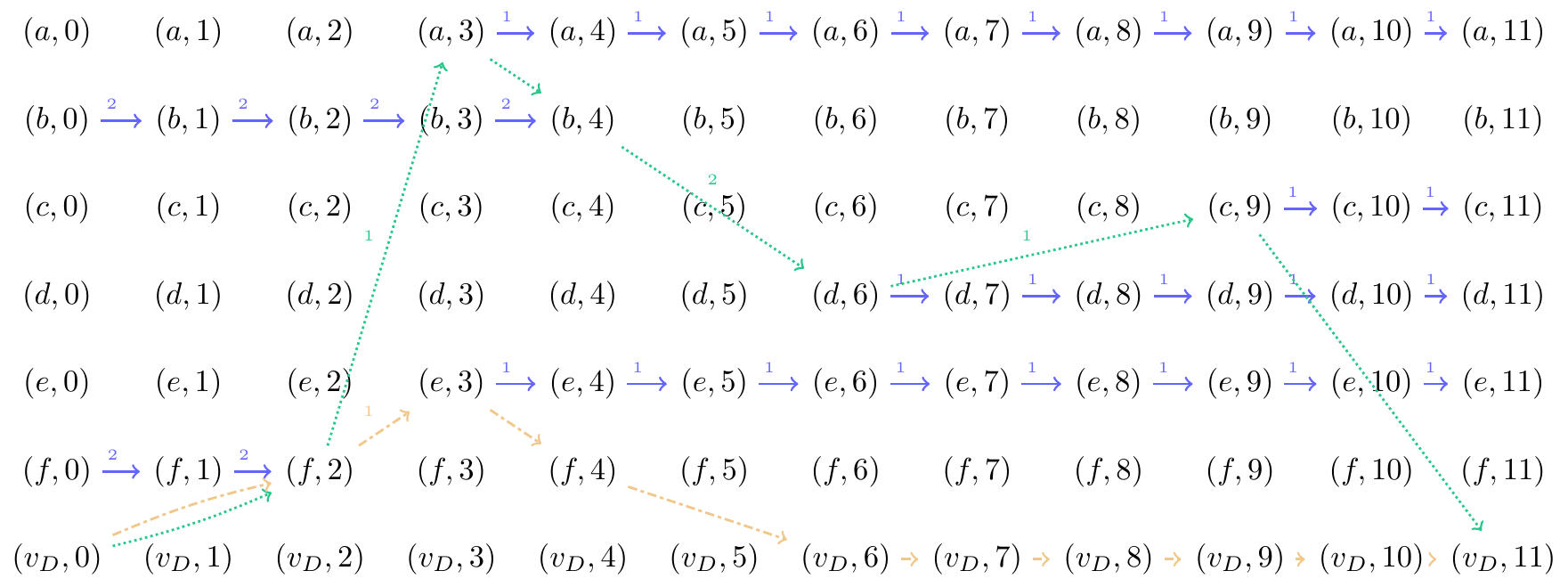}
  \caption{This figure shows the time-expanded network $G_T$ for the graph $G$ from Figure~\ref{fig: reopt}.
    The time horizon is set to $T = 11$.
    Every node of the form~$(v,t)$ represents a station~$v$ (or the depot) at time~$t$.
    In this car sharing system, there are two drivers.
    The driver flow $\fd_1$ is indicated by dash dotted arcs, the driver flow~$\fd_2$ by dotted arcs, and the car flow $\fc$ by solid arcs.
    The car flows~$\fc_1$ and~$\fc_2$ are not shown in the image, however, their flow values are at the corresponding driver flows (if the value is positive).
    The flow values of~$\fc$ on the holdover arcs are the numbers at the holdover arcs.
    The total tour length is 17.}
  \label{fig: flownetwork}
\end{figure}

%%% %%% %%%
\subsection{Integer linear program for Min-Cost-Flow in $G_T$}
\label{sec: static: min-cost flows: ilp: wo pre w back}

To solve the Static Relocation Problem exactly, we aim at determining convoy tours with a minimal total tour length.
For that, we summarize the previous sections by presenting an integer linear programming formulation for a min-cost flow problem in the time-expanded network $G_T = (V_T, A_T)$:
\begin{subequations}\label{eq: static: min-cost flows: ilp: wo pre w back}
  \begin{flalign}
    \min& \,  \sum_{a\in A_L} w(a) \sum_{i = 1}^k \fd_i(a)                                                                                                      \label{eq: static: min-cost flows: ilp: wo pre w back: 1}\\
    & \fc((v, 0), (v, 1)) + \sum_{a\in \delta^-_L(v,0)} \sum_{i=1}^k \fc_i(a) = z_v^0,                           && \forall (v,0) \in V_T                       \label{eq: static: min-cost flows: ilp: wo pre w back: 10}\\
    & \fc((v,T-1),(v,T)) + \sum_{a\in \delta^+(v,T)} \sum_{i=1}^k \fc_i(a) = z_v^T,                              && \forall (v,T) \in V_T                       \label{eq: static: min-cost flows: ilp: wo pre w back: 11}\\
    & \sum_{a\in \delta^-(v, 0)} \fd_i(a) = \bd_v^i,                                                             && \forall 1 \leq i \leq k                     \label{eq: static: min-cost flows: ilp: wo pre w back: 12}\\
    & \sum_{a\in \delta^+(v_D,T)} \fd_i(a) = \bd_v^i,                                                            && \forall 1 \leq i \leq k                     \label{eq: static: min-cost flows: ilp: wo pre w back: 13}\\
    & 
    \begin{aligned}
    \fc((v&, t-1), (v, t)) + \sum_{a\in \delta^+_L(v,t)} \sum_{i=1}^k \fc_i(a) \\
        &= \fc((v, t), (v, t+1)) + \sum_{a \in \delta^-_L(v,t)} \sum_{i=1}^k \fc_i(a),
    \end{aligned}
                                                                                                                &&  \forall v \in V, 0 < t < T                  \label{eq: static: min-cost flows: ilp: wo pre w back: 14}\\
    & \sum_{a\in \delta^+(v,t)} \fd_i(a) = \sum_{a\in \delta^-(v,t)} \fd_i(a),                && \hspace{-1.2cm} \forall v \in V, 0 < t < T, \forall 1 \leq i \leq k    \label{eq: static: min-cost flows: ilp: wo pre w back: 15}\\
    & \sum_{a \in \delta^+_L(v_o, t)} \fc_i(a) \leq \sum_{a \in \delta^-_L(v_o, t)} \fc_i(a), && \hspace{-1.7cm}  \forall v_o \in V_O, \forall 0 < t < T, \forall 1 \leq i \leq k \label{eq: static: min-cost flows: ilp: wo pre w back: 16} \\
    & \sum_{a \in \delta^+_L(v_u, t)} \fc_i(a) \geq \sum_{a \in \delta^-_L(v_u, t)} \fc_i(a), && \hspace{-1.7cm}  \forall v_u \in V_U, \forall 0 < t < T, \forall 1 \leq i \leq k \label{eq: static: min-cost flows: ilp: wo pre w back: 17} \\
    & \sum_{a \in \delta^+_L(v_b, t)} \fc_i(a) = \sum_{a \in \delta^-_L(v_b, t)} \fc_i(a),    && \hspace{-1.7cm}  \forall v_b \in V_B, \forall 0 < t < T, \forall 1 \leq i \leq k \label{eq: static: min-cost flows: ilp: wo pre w back: 18} \\
    & \fc_i(a) \leq \capd \cdot \fd_i(a),                                                        && \forall a \in A_L, \forall 1 \leq i \leq k                  \label{eq: static: min-cost flows: ilp: wo pre w back: 20}\\
    & \fc, \fc_i \text{ integer}, \fd_i \text{ binary},          \label{eq: static: min-cost flows: wo pre w back: ilp: 21}
  \end{flalign}
\end{subequations}
where $\delta^-(v,t)$ denotes the set of outgoing arcs of $(v,t)$ (and $\delta^-_L(v,t)$ denotes the set of outgoing relocation arcs of $(v,t)$),
and $\delta^+(v,t)$ denotes the set of incoming arcs of $(v,t)$ (and $\delta^+_L(v,t)$ denotes the set of incoming relocation arcs of $(v,t)$). 
The objective function~\eqref{eq: static: min-cost flows: ilp: wo pre w back: 1} measures and minimizes the total tour length of the transportation schedule.
The equalities~\eqref{eq: static: min-cost flows: ilp: wo pre w back: 10} and~\eqref{eq: static: min-cost flows: ilp: wo pre w back: 12} give the initial number of cars (resp.~drivers) for each station (resp.~for the depots).
Conditions~\eqref{eq: static: min-cost flows: ilp: wo pre w back: 14} and~\eqref{eq: static: min-cost flows: ilp: wo pre w back: 15} are the flow conservation constraints for the flows~$\fd_i$ and $\fc, \fc_i$, respectively.
The conditions~\eqref{eq: static: min-cost flows: ilp: wo pre w back: 11} and~\eqref{eq: static: min-cost flows: ilp: wo pre w back: 13} ensure to reach the destination state
and that every driver returns to its depot at the end of the time horizon.
The conditions~\eqref{eq: static: min-cost flows: ilp: wo pre w back: 16}, \eqref{eq: static: min-cost flows: ilp: wo pre w back: 17} and~\eqref{eq: static: min-cost flows: ilp: wo pre w back: 18}
ensure that cars are only picked up at overfull stations, and dropped at underfull stations.
Finally, the constraints~\eqref{eq: static: min-cost flows: ilp: wo pre w back: 20} couple the flows~$\fc_i$ and~$\fd_i$ so that cars on relocation arcs cannot move without their driver.
Furthermore, these constraints give the capacities for $\fc_i$ on relocation arcs.

Finally, the flows in the time-expanded network have to be interpreted as a transportation schedule.
Hereby, car and driver flows on relocation arcs correspond to a move, differences in the car flow to actions.
Note, one can easily compute a transportation schedule from the resulting flows.
For that, positive flow $\fd_j$ on a relocation arc $a = ((v, t_v),(v', t_{v'}))$ corresponds to a move $\move = (j, v, t_v, v', t_{v'}, \fc_j(a))$.
Actions are implied by differences of flow values between incoming and outgoing relocation arcs, i.e., if we have~$\fd_j(a) = \fd_j(a') = 1$ with $a' = ((v',t_{v'}), (v'', t_{v''}))$,
and $\fc_j(a) \neq \fc_j(a')$ then the action~$\action = (j, v', t_{v'}, \fc_j(a) - \fc_j(a'))$ is implied.
This implies:

\begin{theorem}
The optimal solution of system \eqref{eq: static: min-cost flows: ilp: wo pre w back} corresponds to a non-preemptive transportation schedule with minimal total tour length
for the Static Relocation Problem with backhaul $(G,\z^0,\z^T, \zd, k,L)$ (Symmetric and Asymmetric).
\end{theorem}

%%% %%% %%%
\subsection{Static Relocation Problem without backhaul}
\label{sec: static: min-cost flows: without backhaul}

In this section, we consider the Static Relocation Problem without backhaul, and
highlight the differences between the exact approach for solving the Static Relocation Problem with and without backhaul.

In the integer linear program stated in the previous section, the equations~\eqref{eq: static: min-cost flows: ilp: wo pre w back: 13} ensure that every driver returns to its depot.
Thus, it is sufficient to conveniently modify these equations.

For that, we replace the equations~\eqref{eq: static: min-cost flows: ilp: wo pre w back: 13} by the following two constraints
\begin{align}
  & \sum_{a\in \delta^+(v_D,T)} \fd_i(a) \leq 1,                                        && \forall 1 \leq i \leq k, \forall v_D \in \VD,                  \label{eq: static: min-cost flows: ilp: wo pre wo back: 1} \\
  & \sum_{a\in \delta^+(v,T)} \fd_i(a) = 0,                                             && \forall 1 \leq i \leq k, \forall v \in V \setminus \VD.        \label{eq: static: min-cost flows: ilp: wo pre wo back: 2}
\end{align}
Hereby, \eqref{eq: static: min-cost flows: ilp: wo pre wo back: 1} ensures that every driver can return to any depot
and \eqref{eq: static: min-cost flows: ilp: wo pre wo back: 2} ensures that the drivers do not end their tours in a station but only in a depot.
Thus, we can compute non-preemptive transportation schedules with a minimal total tour length for the Symmetric (resp.~Asymmetric) Static Relocation Problem without backhaul.

%%%%%%%%%%%%%%%%%%%%%%%%%%%%%%%%%%%%%%%%%%%%%%%%%%%%%%%%%%%%%%%%%%%%%%%%%%%%%%%%%%%%%%%%%%%%%%%%%%%%%%%%%%
\section{The combinatorial algorithm \REOPT}
\label{sec: static: reopt}

The computation times for computing an exact solution by solving the integer linear program~\eqref{eq: static: min-cost flows: ilp: wo pre w back: 10}--\eqref{eq: static: min-cost flows: wo pre w back: ilp: 21}
are extremely high even for small instances (see Table~\ref{tab: computational results}).
This motivates the study of a combinatorial algorithm to solve the problem in reasonably short time.
Therefore, we describe in detail the strategy \REOPT\ proposed in \cite{LAGOS2013} to solve a Static Relocation Problem.
The input of \REOPT\ is the Static Relocation Problem $(M, \z^0, \z^T, \zd, k, \capd)$.
Hereby, we consider a complete weighted graph $G = (\VO \cup \VU \cup \VD, E, d)$ induced by a (quasi) metric space $M = (V, d)$, containing only the overfull stations $\VO \subseteq V$ (with $z_i^0 > z_i^T$),
the underfull stations $\VU \subseteq V$ (with $z_i^0 < z_i^T$), and a set of depots $\VD$, as well as all connections $E$ between them and driving times $d \colon E \to \NN$.
The task set ${\cal T}$ again consists of the tasks $\tau = (v, z_v^0 - z_v^T)$ for all $v \in \VO \cup \VU$.
The output of \REOPT\ is a non-preemptive transportation schedule for the (quasi) metric task system $(M, {\cal T})$, the objective is to minimize its total tour length.

The \REOPT\ approach performs three steps.
Firstly, we construct a weighted complete bipartite graph and find a matching between overfull and underfull stations with minimal edge weight.
Each edge in this matching corresponds to a transport request between two stations.
Secondly, tours for all convoys are constructed (using a heuristic insertion technique) serving each transport request.
Since the transport requests stemming from the minimum matching do not necessarily lead to optimal convoy tours, 
the final step is to iteratively augment the tours by ``rematching'' certain origin/destination pairs, i.e., to reinsert accordingly adapted moves in such a way that the total tour length decreases.

%%%%%%%%%%%%%%%%%%%%%%%%%%%%%%%%%%%%%%%%
\subsection{First step: Compute transport requests}
\label{sec: static: reopt: first step}

In the first step, we compute transport requests of the form $(v_o, v_u, x)$, where $v_o$ is an overfull station, $v_u$ an underfull station and $x$ is the number of cars to be transported from $v_o$ to $v_u$. 
For that, we construct a weighted complete bipartite graph $B = (\VO \cup \VU, A, \boldsymbol{w}, \boldsymbol{p})$ and
consider a restricted vector $\boldsymbol{w} \in \RR^{|A|}$ of edge weights (reflecting the distance between the two adjacent stations)
and a vector $\boldsymbol{p} \in \NN^{|\VO \cup \VU|}$ of node weights reflecting the number $p_v = |z_v^0-z_v^T|$ of cars which have to be moved in or out the corresponding station $v$.

Define a \emph{perfect $p$-matching} in~$B$ to be a multiset $x \colon A \rightarrow \NN$ of the edges such that for each node $v \in \VO \cup \VU$,
exactly $p_v$ incident edges are selected, counted with multiplicities $x_a$. 
Note that by construction of $\boldsymbol{p} \in \NN^{|\VO \cup \VU|}$, the existence of such a perfect $p$-matching
is ensured by $\sum_{v_o \in \VO} p_{o} = \sum_{v_u \in \VU} p_{u}$ since $\sum_{v \in V} z_v^0 = \sum_{v \in V} z_v^T$. 
The goal is to find a perfect $p$-matching~$x$ with minimal edge weight $\sum_{a \in A} w_a x_a$, including multiplicities.
The problem can be formulated by the following integer linear program
\begin{subequations}
\label{eq: static: reopt: matching: ilp}
  \begin{flalign}
    \min& \,  \sum_{a\in A} w_a x_a,                                                                                                                            \label{eq: static: reopt: matching: ilp: 1}\\
    & \sum_{a \in \delta^+(v_o)} x_a = p_{v_o},                                                                 && \forall v_o \in V_O                          \label{eq: static: reopt: matching: ilp: 10}\\
    & \sum_{a \in \delta^+(v_u)} x_a = p_{v_u},                                                                 && \forall v_u \in V_U                          \label{eq: static: reopt: matching: ilp: 11}\\
    & x_a \text{ integer}.                                                                                                                                      \label{eq: static: reopt: matching: ilp: 21}
  \end{flalign}
\end{subequations}
Note that the constraint matrix is totally unimodular, and thus, problem can be solved efficiently~\cite{}.

Each selected matching edge $a = v_ov_u$, with $v_o \in \VO$ and $v_u \in \VU$, corresponds to a transport request for $x_a$ cars from station $v_o$ to station $v_u$.
The set $\TR$ of all such transport requests provides the input for a Pickup and Delivery Problem (PDP) which has to be solved subsequently in order to construct tours for all convoys serving each transport request.

%%%%%%%%%%%%%%%%%%%%%%%%%%%%%%%%%%%%%%%%
\subsection{Second step: Serving the transport requests}
\label{sec: static: reopt: second step}

In this section, we give an algorithm~\PDPINSERT\ which solves the PDP using heuristic insertion techniques.
The input for \PDPINSERT\ is the complete weighted graph $G=(\VO \cup \VU \cup \VD, E, w)$, the total number $k$ of drivers, the convoy capacity $\capd$, and the set of transport requests $\TR$.
The output of \PDPINSERT\ is a non-preemptive transportation schedule for the drivers, which serves all transport requests in~$\TR$ within the time horizon~$[0,T]$.

For that, we define some further notions.
Let $\tourd_j$ be a tour for a driver~$j$ and let $tr = (v_o, v_u, x)$ be a transport request.
We say that $tr$ is \emph{served} by $\tourd_j$ if there exists a pickup action $\action^o = (\cdot, v_o, y)$ and a drop action $\action^u = (\cdot, v_u, -y)$ in $\tourd_j$,
so that $\aexetime(\action^o) < \aexetime(\action^u)$.
Hereby, $tr$ is \emph{fully served} if $y = x$ and \emph{partially served} if $y < x$.
By $y^{tr}(\tourd_j) := y$ we denote the number of cars served from $tr$ by $\tourd_j$. 
A transport request $tr = (v_o, v_u, x)$ can be \emph{(partially) inserted} into a tour~$\tourd_j$, serving~$y$ cars, as follows:
\begin{itemize}
 \item select a move $\move = (j, v, t^v, w, t^w, x_m)$ where $v_o$ shall be inserted,
 \item remove $\move$ from $\tourd_j$,
 \item add move $(j, v, t^v, v_o, t^v + \dist(v, v_o), x_m)$, action $(j, v_o, y)$, move $(j, v_o, t^v + \dist(v, v_o), w, t^v + \dist(v, v_o) + \dist(v_o, w), x_m + y)$ to the tour,
 \item update departure and arrival times of all successive moves and actions,
 \item do analogous steps for $u$.
\end{itemize}
This yields a new tour $\tourd'_j$ which (partially) serves $tr$.
By applying the opposite steps, $tr$ can be \emph{removed} from $\tourd'_j$, which yields $\tourd_j$.

Now let $tr$ be (fully or partially) served by $\tourd$, and $\tourd^w_j$ be the tour derived from $\tourd_j$ without serving $tr$.
We denote the \emph{marginal costs per load unit} $CM(tr, \tourd_j)$ by
\[
 CM(tr, \tourd_j) = \frac{len(\tourd_j) - len(\tourd^w_j)}{y^{tr}(\tourd_j)}.
\]
Now the algorithm \PDPINSERT\ can be described as follows:
\begin{enumerate}
\item For each driver~$j$ initialize the tour so that it starts and ends in the drivers depot~$v_D \in \VD$, i.e., initialize the tour with the move $(j, v_D, 0, v_D, 0, 0)$.
\item Choose a transport request $tr$ that has an origin or a destination already in a tour, else randomly select one.
\item Calculate the marginal cost per load unit of inserting this transport request to each possible tour.
      Hereby, take the number of cars to be transported into account (i.e., respect the convoy capacity) as well as the time.
      Select the tour with the minimum marginal cost per load unit and insert the transport request into this tour.
\item If the transport request $tr$ is fully served, remove it from $\TR$.
      Otherwise, it is partially served by a tour $\tourd$.
      Then subtract the number of cars inserted from the load of the transport request, i.e., remove $tr$ from $\TR$ and add a new transport request $tr' = (v_o, v_u, x - y^{tr}(\tourd))$ to $\TR$.
\item Repeat these steps until all the transport requests are fully served.
\end{enumerate}

%%%%%%%%%%%%%
\begin{figure}[ht]
    \centering
\includegraphics[width=0.4\textwidth]{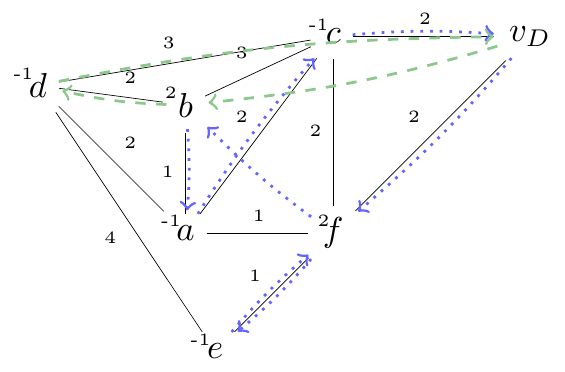} $\, \, \, \, \,$ \includegraphics[width=0.4\textwidth]{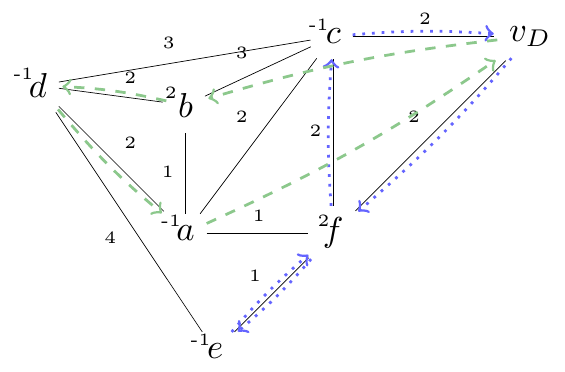}
 \caption{
          In each figure (left and right side), a carsharing system with 6 stations ($a$ to $f$) and one depot ($v_D$) is illustrated as a graph $G$.
          There are two tours (dashed and dotted), one for each of the two drivers (server capacity $\capd = 2$), giving a non-preemptive transportation schedule.
          The numbers at the stations show the amount of cars to be moved from ($> 0$) or to the station ($< 0$).
          On the left side, the transportation schedule before the reoptimization step is shown (total tour length = 22), on the right side, the transportation schedule after this step (total tour length = 19).
 }
 \label{fig: reopt}
\end{figure}
%%%%%%%%%%%%%

The algorithm \PDPINSERT\ computes a non-preemptive transportation schedule. 
However, the transportation schedule created from the transport requests stemming from the minimal perfect $p$-matching does not necessarily lead to optimal tours.
Therefore, the final step is to iteratively augment the tours by ``rematching'' certain origin/destination pairs,
i.e., to reinsert accordingly adapted moves in such a way that the total tour length decreases.

%%%%%%%%%%%%%%%%%%%%%%%%%%%%%%%%%%%%%%%%
\subsection{Third step: Reoptimization}

The algorithm defined here involves the two previous steps: computation of transport requests and the algorithm \PDPINSERT.
The input for the reoptimization step is a transportation schedule $\sched$ serving all transport requests in $\TR$ and two natural numbers $\Delta, N \in \NN$.
The output is a transportation schedule having a total tour length less or equal to the total tour length of $\sched$.
The algorithmic scheme of the reoptimization step is as follows:
\begin{enumerate}
 \item From $\sched$ we withdraw the $N$ transport requests with highest marginal cost per load: $\TR' = \{ (v_o^1, v_u^1, x^1), \dotsc, (v_o^N, v_u^N, x^N) \}$.
 \item From the withdrawn transport requests we compute sets of over- and underfull stations, i.e.,
       the set $V_{O^*} = \{ v_o \in \VO \mid (v_o, \cdot, \cdot) \in \TR' \}$, the set $V_{U^*} = \{ v_u \in \VU \mid (\cdot, v_u, \cdot) \in \TR' \}$ and the vector $\boldsymbol{x} \in \NN^{\abs{V_{O^*}} \cdot \abs{V_{U^*}}}$
       by $x_{ou} = \min\{ p_o, p_u \}$, where $p_o = \sum_{(v_o, \cdot, x) \in \TR'} x$ and $p_u = \sum_{(\cdot, v_u, x) \in \TR'} x$.
 \item For every pair $(v_o, v_u) \in V_{O^*} \times V_{U^*}$ and every tour $\tourd$ we compute the additional marginal cost $CM^+(\tourd, (v_o,v_u))$, where $CM^+(\tourd, (v_o,v_u)) = \tourd^{(v_o,v_u)} - \tourd$
       and $\tourd^{(v_o,v_u)}$ is the tour after inserting a transport request $(v_o, v_u, 1)$.
       Let $w_{ou} = \min_{\tourd \in \sched} CM^+(\tourd, (v_o,v_u))$ be the minimal additional marginal cost.
 \item Next we generate a weighted complete bipartite graph $B = (V_{O^*} \cup V_{U^*}, A, \boldsymbol{w}, \boldsymbol{p})$ and compute a minimal perfect $p$-matching (as in Step 1).
       From the minimal perfect $p$-matching, transport requests are generated, which serve as input for the algorithm \PDPINSERT\ (as in Step 2).
 \item Redo these steps $\Delta$ times.
 \item Finally, we return the best found transportation schedule, i.e., the one with the smallest total tour length.
\end{enumerate}

Figure~\ref{fig: reopt} shows an example for the reoptimization step improving a transportation schedule stemming from the minimal perfect $p$-matching of Step 1.

The algorithm \REOPT\ is summarized in Algorithm~\ref{alg: static: reopt}.

%%%%%%%%%%%%%%%%%%
\begin{algorithm}[ht]
\caption{\REOPT}
\label{alg: static: reopt}
\begin{algorithmic}[1]
  \Require{a Static Relocation Problem $(M, \z^0, \z^T, \zd, k, \capd)$, integers $N$, $\Delta$}
  \Ensure{a non-preemptive transportation schedule}
  \State{Find a minimal perfect $p$-matching (Step 1)}                                                                          \label{alg: static: reopt: 1}
  \While{$counter < \Delta$}                                                                                                    \label{alg: static: reopt: 2}
    \State{update $counter$}                                                                                                    \label{alg: static: reopt: 3}
    \State{construct $k$ tours by \PDPINSERT\ serving all transport requests (Step 2)}                                          \label{alg: static: reopt: 4}
    \State{rematch after withdrawing the N requests that have highest additional marginal costs in their tours (Step 3)}        \label{alg: static: reopt: 5}
  \EndWhile{}
  \State{\Return{transportation schedule of smallest found total tour length}}                                            \label{alg: static: reopt: 6}
\end{algorithmic}
\end{algorithm}
%%%%%%%%%%%%%%%%%%

Finally, we give some comments about the complexity of the algorithm \REOPT.
The minimal perfect $p$-matching of Step~1 of the algorithm \REOPT\ (Section~\ref{sec: static: reopt: first step}), can be computed in polynomial time.
Since constructing an optimal transportation schedule from the minimal perfect $p$-matching (Section~\ref{sec: static: reopt: second step}) results in a dial-a-ride problem, this step is at least~\NPhard.
Thus, the second step cannot be solved in polynomial time unless $\mathcal{P} = \mathcal{NP}$ holds.
Therefore, the total complexity of \REOPT\ is at least in~$\mathcal{NP}$.

%%%%%%%%%%%%%%%%%%%%%%%%%%%%%%%%%%%%%%%%
\section{Approximation factors}
\label{sec: static: reopt: approximation factor}

In this section, we consider several different situations and show that in most of these cases, the algorithm \REOPT achieves a finite approximation factor based on the capacity of the convoys.
We consider the symmetric and asymmetric situations, when there is only one depot on the system and when there are multiple depots.
In the case of multiple depots, we distinguish between tours with and without ``backhaul'', i.e., whether each driver has to return to its starting depot (with backhaul) or to any depot in the system (without backhaul).
Obviously, in the case of only a single depot in the system, the Static Relocation Problem with backhaul and without backhaul coincide.

Firstly, we consider the case when there is only one depot in the system (Section~\ref{sec: static: reopt: single depot}).
Hereby, we show that the approximation factor is equal in both cases, the symmetric and the asymmetric one.
Secondly, we consider multiple depots in the system (Section~\ref{sec: static: reopt: multiple depot}).
In this case, the symmetric and asymmetric situations become different, depending whether the transportation schedule contains only tours with or without backhaul.

Throughout this section, we assume that the time horizon is always large enough.

%%%%%%%%%%%%%%%%%%%%%%%%%%%%%%%%%%%%
\subsection{Single depot}
\label{sec: static: reopt: single depot}

%%%%%%%%%%%%%%%%%%%%%%%%%%%%%%%%%%%%
In this section, we show that in the case when there is only one depot in the system, \REOPT achieves a finite approximation factor based on the capacity of the convoys.

%%%%%%%%%%%%%%%%%%
\begin{theorem}\label{thm: static: reopt: approximation factor}
For the Static Relocation Problem $(M, \z^0, \z^T, \zd, k, \capd)$ with one depot,
the algorithm \REOPT\ %computes a non-preemptive transportation schedule and 
achieves an approximation factor of $\capd + 1$. % for all $\capd \in \NN$.
\end{theorem}
%%%%%%%%%%%%%%%%%%

That \REOPT\ computes a non-preemptive transportation schedule for the Static Relocation Problem has already been observed in the previous section.
In order to prove the approximation factor, we first introduce some definitions, as well as state and prove some lemmas.

Firstly, from a given tour, we construct a new tour where each action picks up (resp.~drops) exactly one car.
Considering such tours only simplifies several technical issues, like estimating the number of consecutive pickup actions.
Secondly, we construct a new transportation schedule from an optimal transportation schedule and a minimal perfect $p$-matching.
Finally, we compare the lengths of an optimal transportation schedule~$\sched^*$, a transportation schedule $\sched^p$ derived from a minimal perfect $p$-matching, and the constructed transportation schedule $\sched^t$ and show that
\[
  \ell(\sched^p) \leq \ell(\sched^t) \leq (L + 1) \ell(\sched^*)
\]
holds, which proves the stated approximation factor.
Hereby, we construct the transportation schedule~$\sched^p$ by taking moves and actions of the optimal transportation schedule~$\sched^*$ and
by constructing moves from the transport requests of the minimal perfect $p$-matching.
Then, the approximation factor $(L + 1)$ emerges from the maximal number of moves corresponding to moves through the system in order to pickup cars and from moves which are serving a transport requests.

We start by showing how to construct a tour $\overline{\tourd}$ from a given tour $\tourd = (\move_1, \action_1, \dotsc, \action_{\ntourd - 1}, \move_{\ntourd})$, where in each action exactly one car is picked up or dropped.

For every $i \in \{ 1, \dotsc, \ntourd \}$ do
\begin{itemize}
 \item add move $\move_i$ to $\overline{\tourd}$,
 \item for every action~$\action_i = (j, v, x)$ where more than one car is picked up from a station~$v$ at time~$t_v$, we ``replace'' the action by actions each picking one car and waiting moves (with 0 waiting time) between these actions,
       i.e., if $x > 1$ then add the following $x$ actions and $x-1$ moves 
       $((j, v, 1), (j, v, t_v, v, t_v, x_{i} + 1), \dotsc, (j, v, t_v, v, t_v, x_{i} + x - 1), (j, v, 1))$ are added to $\overline{\tourd}$,
 \item for every action~$\action_i = (j, v, x)$ where more than one car is dropped at a station~$w$ at time~$t_w$, we ``replace'' the action by actions each picking one car and waiting moves (with 0 waiting time) between these actions,
       i.e., if $x < -1$ then add the following $x$ actions and $x-1$ moves 
       $((j, w, -1), (j, w, t_w, w, t_w, x_{i} - 1), \dotsc, (j, w, t_w, w, t_w, x_{i} - x  + 1), (j, w, -1))$ are added to $\overline{\tourd}$,
 \item every action $\action^i = (j, v, x)$ with $-1 \leq x \leq 1$ is added unchanged to $\overline{\tourd}$.
\end{itemize}
The tour $\overline{\tourd}$ is called a \emph{uniform tour} corresponding to $\tourd$.
A transportation schedule containing only uniform tours is called \emph{uniform transportation schedule}.

Note that a uniform tour is indeed a tour.
Furthermore, note that there exists exactly one uniform tour corresponding to a tour (if no unnecessary empty actions and waiting moves are added), %but to a uniform tour, there can be several corresponding tours.
but from a uniform tour, one can generally derive several non-uniform tours.

In this section, we consider non-preemptive transportation schedules, i.e., there does not exist a tour depending on another tour.
Therefore, empty actions can be safely removed from any tour in a transportation schedule (some moves may need to be adjusted accordingly).
For the rest of this section, we assume that no action is empty.

%%%%%%%%%%%%%%%%%%
\begin{example}\label{ex: static: reopt: uniform tour}
Let us consider the graph and the dashed tour for driver~$1$ from the right side of Figure~\ref{fig: reopt}.
The tour is then given by
\begin{align*}
 \tourd = \{& \\
            & (1, v_D, 0, b, 4, 0), \\
            & \mathbf{(1, b, 2)}, \\
            & (1, b, 4, d, 6, 2), \\
            & (1, d, -1), \\
            & (1, d, 6, a, 8, 1), \\
            & (1, a, -1), \\
            & (1, a, 8, v_D, 11, 0) \\
          \},&
\end{align*}
and the corresponding uniform tour is then
\begin{align*}
 \overline{\tourd} = \{& \\
            & (1, v_D, 0, b, 4, 0), \\
            & \mathbf{(1, b, 1)}, \\
            & \mathbf{(1, b, 4, b, 4, 1)}, \\
            & \mathbf{(1, b, 1)}, \\
            & (1, b, 4, d, 6, 2), \\
            & (1, d, -1), \\
            & (1, d, 6, a, 8, 1), \\
            & (1, a, -1), \\
            & (1, a, 8, v_D, 11, 0) \\
          \}.&
\end{align*}
The ``replaced'' action is highlighted with bold fonts, all other actions $a \in \tact(\tourd)$ have already the form $\acnum(a) = \pm 1$.
%%%%%%%%%%%%%%%%%%
\end{example}
%%%%%%%%%%%%%%%%%%

The following lemma is a direct conclusion from the construction of an uniform tour.

%%%%%%%%%%%%%%%%%%
\begin{lemma}\label{lem: static: reopt: max consecutive actions}
  Let $\overline{\tourd} = (\move_1, \action_1, \dotsc, \move_{\ntourd - 1}, \action_{\ntourd - 1}, \move_{\ntourd})$ be a uniform tour for driver~$j$.
  Then there are at most $\capd$ consecutive pickup (resp.~drop) actions in the sequence $\tact(\overline{\tourd})$.
\end{lemma}

%%%%%%%%%%%%%%%%%%

Next, we construct a graph~$G$ from a given tour, where the set of nodes corresponds to the actions, and the set of arcs to the moves of the tour.
Afterwards, we combine this graph with transport requests (leading to a graph~$G^t$), which then helps us to construct another tour (from this tour we finally gain the transportation schedule~$\sched^t$).
This constructed tour has some nice properties with respect to the number of traverses of each arc of~$G^t$, which finally helps us to prove our main result (Theorem~\ref{thm: static: reopt: approximation factor}).

From a given tour $\tourd = (\move_1, \action_1, \move_2, \dotsc, \action_{\ntourd-1}, \move_{\ntourd})$, we construct a directed weighted graph $G = (\Vpick \cup \Vdrop \cup \Vbal, A, w)$, where
\begin{enumerate}
 \item the set of nodes $V = \Vpick \cup \Vdrop \cup \Vbal$ corresponds to the actions in $\tourd$, $\Vpick$ corresponds to the set of pickup actions, $\Vdrop$ to the set of drop actions, and $\Vbal$ to an empty action at the depot $v_D$, i.e., $\Vbal = \{ (\cdot, v_D, 0) \}$;
% \item the set of arcs $A$ corresponds to the moves, connecting the actions;
 \item there is an arc from $v \in V$ to $v' \in V$ if $v = \action_j$ and $v' = \action_{j+1}$ for a $1 \leq j \leq \ntourd$, furthermore there is an arc from $(\cdot, v_D, 0)$ to $\action_1$ and from $\action_{\ntourd - 1}$ to $(\cdot, v_D, 0)$;
 \item the weight function $w$ corresponds to the distances between the origin and destination stations of the corresponding moves, i.e., we set $w(\action_j, \action_{j+1}) = d(\orig(\move_{j+1}), \dest(\move_{j+1}))$.
\end{enumerate}
We call such a graph a \emph{tour graph} for $\tourd$, the set $A$ is called the set of \emph{tour arcs}.

Note that one can assign to every arc $a \in A$ of a tour graph a move $\move \in \tourd$.
Then $\move$ is called the \emph{corresponding move} to $a$.

%%%%%%%%%%%%%%%%%%
\begin{example}
The tour graph for the tour $\overline{\tourd}$ from Example~\ref{ex: static: reopt: uniform tour} is illustrated in Figure~\ref{fig: static: reopt: tour 1 action}.
%%%%%%%%%%%%%%%%%%
\begin{figure}[ht]
    \centering
    \includegraphics[width=0.4\textwidth]{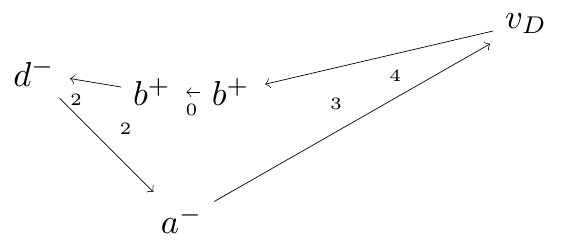}
 \caption{
  This figure shows the tour graph for the uniform tour $\overline{\tourd}$ from Example~\ref{ex: static: reopt: uniform tour}.
  The weights of the arcs correspond to the shortest distance between the connected two stations.
  In the tour graph, a pickup action at station $v \in V$ is denoted by $v^+$ and a drop action at $v$ by $v^-$.
 }
 \label{fig: static: reopt: tour 1 action}
\end{figure}
\end{example}
%%%%%%%%%%%%%%%%%%

Analogously, to a uniform tour we now define a set of uniform transport requests.
A transport request $(v_i, w_i, x_i)$ is called \emph{uniform} if~$x_i = 1$.
Obviously, every set of transport requests can be transformed into a set of uniform transport requests by splitting every transport request $(v_i, w_i, x_i)$ into~$x_i$ uniform transport requests.
A set of transport requests $\TR$ is called \emph{set of uniform transport requests} if every transport request $r \in \TR$ is uniform.

To each uniform transport request $(v, w, 1) \in \TR$, we can now assign two actions for a driver~$j$, one pickup $(j, v, 1)$ and one drop action $(j, w, -1)$.
Hereby, every action is assigned to exactly one transport request.

Let $G = (V, A, w)$ be a tour graph and let $\TR$ be a set of transport requests.
Then we construct a directed weighted graph $G^t = (V, A \cup A^t, w)$, where
$A^t$ is the set of \emph{transport arcs}, which consist of arcs corresponding to the transport requests in $\TR$,
i.e., for a transport request $r \in \TR$ there is an arc between $v, v' \in V$ if $v$ is the assigned pickup action of $r$ and $v'$ the assigned drop action of $r$.
The weight of a transport arc is equal to the distance between the locations of the two assigned actions.
The graph~$G^t$ is called a \emph{transport graph} (see Figure~\ref{fig: static: reopt: transport graph} for an illustration).

\begin{figure}[ht]
    \centering
    \includegraphics[width=0.4\textwidth]{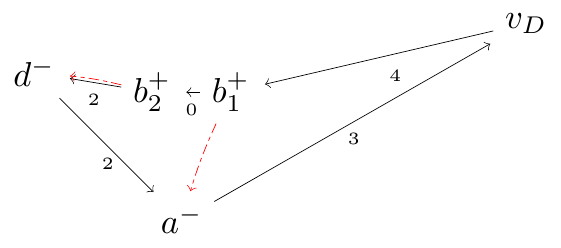}
 \caption{
  This figure shows a transport graph for the uniform tour $\overline{\tourd}$ from Example~\ref{ex: static: reopt: uniform tour} and the set of transport requests
  $\TR = \{ (b, d, 1), (b, a, 1) \}$.
  The weights of the arcs correspond to the shortest distance between the connected two stations.
  In the tour graph, a pickup action at station $v \in V$ is of the form $v^+$ and a drop action is of the form $v^-$.
  The dash-dotted arcs correspond to transport arcs.
 }
 \label{fig: static: reopt: transport graph}
\end{figure}

%%%%%%%%%%%%%%%%%%
\begin{remark}\label{rem: static: reopt: construct tour from transport graph}
Let $G = (V, A, w)$ be a tour graph for a tour $\overline{\tourd}$ for driver~$j$ and let $\TR$ be a set of transport requests.
From a transport graph $G^t = (V, A \cup A^t, w)$ for $G$ and $\TR$, one can construct a new tour $\overline{\tourd}^t$ that serves all transport requests in $\TR$, as follows:
\begin{itemize}
 \item Start in the depot $v_D$ (resp.~the node corresponding to the depot).
 \item Consider the next tour arc $a \in A$ or non-traversed transport arc $a^t \in A^t$.
 \item If a tour arc $a = (v, w)$ is selected, we add a corresponding move~$m$ to~$\overline{\tourd}^t$ from $\aloc(v)$ to $\aloc(w)$ with $\loadd(m) = 0$.
 \item If necessary, add empty actions (or merge the moves).
 \item If a transport arc $a^t = (v, w)$ is selected, we add the pickup action~$v$, a move $(j, \aloc(v), \cdot, \aloc(w), \cdot, 1)$ and a drop action~$w$ to~$\overline{\tourd}^t$.
 \item When all transport requests are served, return to depot the by following tour arcs until the depot is reached.
\end{itemize}
The departure and arrival times of a move $\move_j$ are directly induces by the departure and arrival times of the moves preceding $\move_1, \dotsc, \move_{j-1}$.

Note that following this construction, we always construct a tour.
However, without further restrictions
this easy construction
does not ensure an upper bound on the number of traverses of an arc
(later in Algorithm~\ref{alg: static: reopt: construct tour} we give a refined construction which ensures an upper bound on the number of traverses of an arc).
\end{remark}
%%%%%%%%%%%%%%%%%%

When we speak about a \emph{constructed tour} (from a transport graph~$G^t$ and a set of transport requests~$TR$),
we mean a tour ${\tourd}$ which is constructed using only the arcs from~$G^t$ and which serves all transport requests from~$\TR$.

%%%%%%%%%%%%%%%%%%
\begin{example}\label{ex: static: reopt: tour from transport graph}
Let us consider the transport graph from Figure~\ref{fig: static: reopt: transport graph}.
A possible new tour serving all transport requests $\TR$ constructed from the transport graph, is then given by (see Figure~\ref{fig: static: reopt: tour from transport graph} for an illustration)
\begin{align*}
 \overline{\tourd}^t = \{& \\
            & (1, v_D, 0, b_1, 4, 0), \\
            & (1, b_1, 0), \\
            & (1, b_1, 0, b_2, 4, 0), \\
            & (1, b_2, 1), \\
            & (1, b_2, 4, d, 6, 1), \\
            & (1, d, -1), \\
            & (1, d, 6, a, 8, 0), \\
            & (1, a, 0), \\ 
            & (1, a, 8, v_D, 11, 0), \\
            & (1, v_D, 0), \\ 
            & (1, v_D, 11, b_1, 15, 0), \\
            & (1, b_1, 1), \\
            & (1, b_1, 15, a, 17, 1), \\
            & (1, a, -1), \\
            & (1, a, 17, v_D, 19, 0), \\
          \}.&
\end{align*}
Hereby, the stations $b_1$ and $b_2$ both correspond to the station~$b$.
However, for the sake of readability, we use $b_1$ and $b_2$ instead of~$b$.
%%%% %%%% %%%%
\begin{figure}[ht]
    \centering
    \includegraphics[width=0.4\textwidth]{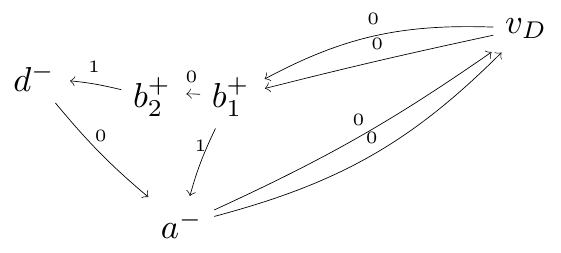}
 \caption{
  This figure shows a possible tour constructed from the transport graph of Figure~\ref{fig: static: reopt: transport graph} serving the transport requests $\TR = \{ (b, d, 1), (b, a, 1) \}$.
 }
 \label{fig: static: reopt: tour from transport graph}
\end{figure}
%%%% %%%% %%%%
\end{example}
%%%%%%%%%%%%%%%%%%

Our goal is to construct a new tour from a transport graph, constructed from an optimal tour, and from a set of transport requests, which is generated from a minimal perfect p-matching.
Then, we show that the approximation factors from Theorem~\ref{thm: static: reopt: approximation factor} hold for this constructed new tour.
For that we define a function which returns for each arc of a transport graph the number of traverses of the arc during the construction of the tour.

Let $\tourd$ be a uniform tour, $\TR$ be a set of transport requests and $G^t = (V, A \cup A^t, w)$ be a transport graph for $\tourd$ and $\TR$.
Furthermore, let $\tourd^t$ be a constructed tour from~$G^t$ and~$\TR$.
Then, we consider a so-called \emph{traverse counter function}, which is a function~$f^a : A \to \NN$ which reflects how often a tour arc $a \in A$ is traversed during the construction of $\tourd^t$.

With the function $f^a$, we counts the traverses of the tour arcs only, while we the request arcs are handled later.

%%%%%%%%%%%%%%%%%%
\begin{example}
\label{ex: static: reopt: traverse counter function}
Let us consider the transport graph from Figure~\ref{fig: static: reopt: transport graph} 
and the tour $\overline{\tourd}^t$ from Example~\ref{ex: static: reopt: tour from transport graph}.
The traverse counter function $f^a$ is then
\begin{align*}
 f^a(v_D, b^+_1)   & = 2, \\
 f^a(b^+_1, b^+_2) & = 1, \\
 f^a(b^+_2, d^-)   & = 0, \\
 f^a(d^-, a^-)     & = 1, \\
 f^a(a^-, v_D)     & = 2. \\
\end{align*}
Note that there are two arcs in the transport graph from Figure~\ref{fig: static: reopt: transport graph} between $b_2^+$ and $d^-$, one tour arc and one request arc.
Hereby, the tour $\overline{\tourd}^t$ is constructed not by traversing the tour arc $(b^+_2, d^-)$
but instead by traversing the transport arc from~$b^+_2$ to~$d^-$ is traversed.
Thus, we have $f^a(b^+_2, d^-) = 0$.
\end{example}
%%%%%%%%%%%%%%%%%%

%%%%%%%%%%%%%%%%%%
\begin{lemma}
\label{lem: static: reopt: f leq c plus 1}
  Let ${\tourd}$ be a uniform tour starting and ending in depot~$v_D$, $\TR$ a set of transport requests, and $G^t = (\Vpick \cup \Vdrop \cup \Vbal, A \cup A^t, w)$ a transport graph for $\tourd$ and $\TR$.
  Furthermore, let $\overline{\tourd}$ be the constructed tour from Algorithm~\ref{alg: static: reopt: construct tour} and let $f^a$ be a traverse counter function for $\overline{\tourd}$.
  Then $f^a(a) \leq \capd + 1$ holds for all tour arcs $a \in A$.
  More specifically, we have
  \begin{enumerate}
   \item \label{lem: static: reopt: f leq c plus 1: 1} $f^a(a) \leq \capd + 1$ holds for all tour arcs $a = (v,w) \in A$ with $v \in \Vdrop$,
   \item \label{lem: static: reopt: f leq c plus 1: 2} $f^a(a') \leq \capd$ holds for all tour arcs $a' = (v', w') \in A$ with $v' \in \Vpick$.
  \end{enumerate}
\end{lemma}

%%%%%%%%%%%%%%%%%%

\begin{proof}

In the Example~\ref{ex: static: reopt: traverse counter function}, we can already make an important observation,
namely, nodes corresponding to a pickup action increase the number of traverses, nodes corresponding to a drop action decrease the number of traverses.
Next, we show that this observation is always true.

%%%%%%%%%%%%%%%%%%
\begin{sublemma}
\label{lem: static: reopt: abs diff 1}
Let $a_1, a_2 \in A$ be two tour arcs with $a_1 = (v, w)$ and $a_2 = (w, u)$ (i.e., $a_2$ is the successive tour arc of $a_1$).
Then $\abs{f^a(a_1) - f^a(a_2)} = 1$ holds.
More specifically, it holds
\begin{enumerate}
 \item \label{lem: static: reopt: abs diff 1: 1} $f^a(a_1) - f^a(a_2) = 1$ if $w \in \Vpick$,
 \item \label{lem: static: reopt: abs diff 1: 2} $f^a(a_1) - f^a(a_2) = -1$ if $w \in \Vdrop$.
\end{enumerate}
\end{sublemma}

%%%%%%%%%%%%%%%%%%

\begin{subproof}
The movement of a driver can be modeled by a flow $f : A \cup A^t \to \NN$.
Then the flow conservation constraint $\sum_{a \in \delta^+(v)} f(a) = \sum_{a \in \delta^-(v)} f(a)$ must hold.
Hereby, the transport arcs must be taken into consideration as well.
It is easy to see that there always exists a flow $f$ with $f(a) = f^a(a)$ for all tour arcs $a \in A$.
Since every node corresponding to a pickup action has exactly one outgoing transport arc and every drop action has exactly one incoming transport arc,
and due to the flow conservation constraint it follows that
\begin{enumerate}
 \item $f^a(a_1) - f^a(a_2) = 1$ if $w \in \Vpick$,
 \item $f^a(a_1) - f^a(a_2) = -1$ if $w \in \Vdrop$
\end{enumerate}
holds and, thus, proves the statement.
\end{subproof}
%%%%%%%%%%%%%%%%%%

There are several consequences from this lemma.
Firstly, we show in the next corollary that the difference of the number of traverses of two different tour arcs can be bounded from above.
Secondly, we show a relation between the number of traverses of an arc and the number of cars that are transfered in the corresponding move within the constructed tour.

%%%%%%%%%%%%%%%%%%
\begin{sublemma}
\label{cor: static: reopt: leq C}
Let $a_1, \dotsc, a_\tau \in A$ be successive tour arcs (i.e, $a_1 = (v_1, v_2), a_2 = (v_2, v_3), \dotsc, a_\tau = (v_\tau, v_{\tau+1})$), and
let $t : A \to \NN$ be a function that returns the number of cars transfered in the corresponding moves in $\overline{\tourd}$.
Then
\begin{enumerate}
 \item \label{cor: static: reopt: leq C: 1} $\abs{f^a(a_1) - f^a(a_\tau)} \leq \capd$, and
 \item \label{cor: static: reopt: leq C: 2} $f^a(a) + t(a) = {const}$ holds for all tour arcs~$a \in A$.
  Especially it holds $f^a(a) + t(a) = f^a(a') + t(a')$ for every pair of tour arcs $a, a' \in A$.
\end{enumerate}
\end{sublemma}

%%%%%%%%%%%%%%%%%%

\begin{subproof}
``\ref{cor: static: reopt: leq C: 1}''
Lemma~\ref{lem: static: reopt: max consecutive actions} shows that in every tour there are maximal $\capd$ consecutive pickup actions and maximal $\capd$ consecutive drop actions.
Furthermore, the number of cars in a convoy must not exceed the capacity~$\capd$.
Since $\abs{f^a(a_j) - f^a(a_{j + 1})} = 1$ holds, the difference cannot be greater than~$\capd$.

``\ref{cor: static: reopt: leq C: 2}''
Let $a \in A$ be a tour arc and let $a' \in A$ be the successive tour arc.
We proof the statement with two cases, when the end node of $a$ is a node in $\Vpick$ and when it is a node in $\Vdrop$.

\textit{Case i} (the end node of $a$ is a node in $\Vpick$):
From Claim~\ref{lem: static: reopt: abs diff 1}~\ref{lem: static: reopt: abs diff 1: 1} it follows that $f^a(a) - f^a(a') = 1$.
Since the end node of $a$ corresponds to a pickup action, the number of cars transported in~$\overline{\tourd}$ is increased by $1$, i.e., we have $t(a) - t(a') = -1$.
Thus, it follows $f^a(a) - f^a(a') + t(a) - t(a') = 0$ and, therefore, $f^a(a) + t(a) = f^a(a') + t(a')$.

\textit{Case ii} (the end node of $a$ is a node in $\Vdrop$):
From Claim~\ref{lem: static: reopt: abs diff 1}~\ref{lem: static: reopt: abs diff 1: 2} it follows that $f^a(a) - f^a(a') = -1$.
Since the end node of $a$ corresponds to a drop action, the number of cars transported in~$\overline{\tourd}$ is decreased by $1$, i.e., we have $t(a) - t(a') = 1$.
Thus, it follows $f^a(a) - f^a(a') + t(a) - t(a') = 0$ and, therefore, $f^a(a) + t(a) = f^a(a') + t(a')$.

We have shown, that the statement holds for two successive arcs.
The statement for arbitrary pairs now follows by iteratively applying the two cases.
\end{subproof}
%%%%%%%%%%%%%%%%%%

Next, we show that the maximum of a traverse counting function is always on the incoming and outgoing tour arcs of the depot.

%%%%%%%%%%%%%%%%%%
\begin{sublemma}
\label{cor: static: reopt: f geq f and f gt f}
  Let $\Vbal = \{ v_D^= \}$.
  Then 
  \begin{enumerate}
  \item $f^a((v_D^=, \cdot)) = f^a((\cdot, v_D^=)) \geq f^a(a')$ for all tour arcs~$a' \in A$,
  \item $f^a((v_D^=, \cdot)) = f^a((\cdot, v_D^=)) > f^a(a')$ for all tour arcs~$a'' = (v, w) \in A$, with $v \in \Vpick$.
  \end{enumerate}
\end{sublemma}

%%%%%%%%%%%%%%%%%%

\begin{subproof}
Let $t : A \to \NN$ be a function that returns the number of cars transfered in the corresponding moves in $\overline{\tourd}$.
From Lemma~\ref{lem: static: reopt: f leq c plus 1} \ref{lem: static: reopt: f leq c plus 1: 2} we know that $f^a(a) + t(a) = f^a(a') + t(a')$ for $a, a' \in A$.
At the beginning and end of a tour, the number of cars in a convoy is always~$0$ and, thus, we have $f^a((v_D^=, \cdot)) = f^a((\cdot, v_D^=)) = f^a(a) + t(a)$ for all $a \in A$.
The statements now follows since $t(a') \geq 0$ for all $a' \in A$ and $t(a'') \geq 1$ for all $a'' = (v, w) \in A$, with $v \in \Vpick$.
\end{subproof}
%%%%%%%%%%%%%%%%%%

In Algorithm~\ref{alg: static: reopt: construct tour} we describe a specific construction for new tours from a given tour and a given set of transport requests.
For this new tour $f^a(a) \leq \capd + 1$ holds for all $a \in A$.
This can be seen as follows.

%%%%%%%%%%%%%%%%%%
\begin{algorithm}[ht]
\caption{Construct new tour}
\label{alg: static: reopt: construct tour}
\begin{algorithmic}[1]
  \Require{a uniform tour $\overline{\tourd} = (\move_1, \action_1, \dotsc, \action_{\nu - 1}, \move_\nu)$ with backhaul for driver~$j$, the depot~$v_D$, a set of transport requests~$\TR$}
  \Ensure{a tour $\tourd$ serving all transport requests of~$\TR$}
  \State{construct transport graph $G^t = (\Vpick \cup \Vdrop \cup \Vbal, A \cup A^t, w)$}
  \State{initialize ${\tourd} \gets \emptyset$}
%   \State{initialize $\overline{V}^+ \gets \emptyset$}                                                                                                              \Comment{contains visited nodes from $\Vpick$}
  \State{initialize $currNode \gets v_D^=$}                                                          \label{alg: static: reopt: construct tour: 4}
  \While{not every node in $\Vpick$ has been visited}
    \If{$currNode \in \Vpick$ \textbf{and} has not been visited}
      \State{mark $currNode$ as visited}
      \State{follow transport arc and add corresponding moves and actions to ${\tourd}$}
    \Else{}
      \State{follow tour arc and add corresponding move to ${\tourd}$}
    \EndIf{}
    \State{update $currNode$}
  \EndWhile{}
  \State{follow tour $\overline{\tourd}$ until arriving in depot node}
  \State{if necessary insert empty actions between two successive moves in ${\tourd}$}
  \State{\Return ${\tourd}$}
\end{algorithmic}
\end{algorithm}
%%%%%%%%%%%%%%%%%%

It is fairly easy to see that Algorithm~\ref{alg: static: reopt: construct tour} follows basically the same steps as in Remark~\ref{rem: static: reopt: construct tour from transport graph}.
Furthermore, one can easily see that really a tour is constructed, serving all transport requests.
In contrast to the construction in Remark~\ref{rem: static: reopt: construct tour from transport graph}, Algorithm~\ref{alg: static: reopt: construct tour} ``follows'' the given tour only until it arrives at a non-visited node corresponding to a pickup action and directly serve the transport request.

Since the constructed tour in Algorithm~\ref{alg: static: reopt: construct tour} ``follows'' the moves from~$\tourd$ until the next non-visited pickup action,
it is ensured, that the number of traverses of an arc is not artificially increased.
Thus, the result follows with the help of Lemma~\ref{lem: static: reopt: max consecutive actions}, Claim~\ref{lem: static: reopt: abs diff 1}, Claim~\ref{cor: static: reopt: leq C} and Claim~\ref{cor: static: reopt: f geq f and f gt f}.
\end{proof}

%%%%%%%%%%%%%%%%%%
\begin{remark} \label{rem: static: reopt: arbitrary start node}
  Although Algorithm~\ref{alg: static: reopt: construct tour} starts the construction of the tour from the depot, any arbitrary node could be used as a starting node.
  That could be done by removing line~\ref{alg: static: reopt: construct tour: 4} and giving $currNode$ as a parameter.
  When the construction is started with another node, Lemma~\ref{lem: static: reopt: f leq c plus 1} still holds, and so %do the following lemmas and corollaries.
  does Lemma~\ref{cor: static: reopt: existence tour}.
\end{remark}
%%%%%%%%%%%%%%%%%%

%%%%%%%%%%%%%%%%%%%%%%%%%%

Let $\overline{\tourd}$ be a uniform tour, and let $\TR$ be a set of uniform transport requests so that $\overline{\tourd}$ serves all transport requests in $\TR$.
Furthermore, let $G^t = (\Vpick \cup \Vdrop \cup \Vbal, A \cup A^t)$ be a transport graph for $\overline{\tourd}$ and $\TR$.
Since every transport request $r = (v, v', 1) \in \TR$ is served by $\overline{\tourd}$, there exists a path of tour arcs in $G^t$ from the action $\action_1$ corresponding to~$v$ to the action $\action_l$ corresponding to~$v'$.
Let this path of tour arcs be $p(\action_1, \action_l) = (\action_1, \action_2), (\action_2, \action_3), \dotsc, (\action_{l-1}, \action_l)$.
Due to the triangle inequality we can estimate the length $w(\action_1, \action_l)$ of the transport arc by
\[
  w(\action_1, \action_l) \leq w(\action_1, \action_2) + \dotsm + w(\action_{l-1}, \action_l).
\]
Therefore, the length of all transport arcs can be estimated by applying above formula iteratively on all transport arcs,
\begin{equation}\label{eq: static: reopt: estimate length all transport arcs}
 \sum_{a \in A^t} w(a) \leq \sum_{a \in A^t} \sum_{(\action, \action') \in p(a)} w(\action, \action'),
\end{equation}
where $p(a)$ is the minimal path of tour arcs in $G^t$ between the corresponding actions.
Now we can consider a function $f^{\TR} : A \to \NN$, called \emph{transport estimate function},
where $f^{\TR}(a)$ shows how often the tour arc $a$ is used in the right hand side of Equation~\eqref{eq: static: reopt: estimate length all transport arcs}.
With~$F^{\TR}$ can reformulate Equation~\eqref{eq: static: reopt: estimate length all transport arcs} as
\begin{equation}\label{eq: static: reopt: estimate length all transport arcs: fTR}
 \sum_{a \in A^t} w(a) \leq \sum_{a \in A^t} \sum_{(\action, \action') \in p(a)} w(\action, \action') = \sum_{a \in A} f^{\TR}(a) w(a).
\end{equation}

It is easy to see, that the values of a transport estimate function depends on the choice of the set of transport requests~$\TR$.
In order to prove the main theorem, we consider a specific set of transport requests.
For that let $\overline{\tourd}$ be a uniform tour for a driver starting and ending in depot~$v_D$.
A set of uniform transport requests~$\TR$ so that
\begin{enumerate}
  \item \label{lem: static: reopt: existence tour: i}   $\overline{\tourd}$ serves all transport requests in~$\TR$,
  \item \label{lem: static: reopt: existence tour: ii}  for every transport request $r \in \TR$ there are at most~$\capd$ actions between two corresponding actions for $r$, and
  \item \label{lem: static: reopt: existence tour: iii} there does not exist a transport request $(\aloc(v), \aloc(w), 1) \in \TR$ so that the minimal path
                                                        from $\aloc(v)$ to $\aloc(w)$ of tour arcs traverses the tour arcs connecting the depot,
\end{enumerate}
holds is called a set of \emph{close distance uniform transport requests} for~$\overline{\tourd}$.

In the next lemma, we show that for every uniform tour there exists a set of close distance uniform transport requests.
Hereby, also the motivation for the choice of the name becomes clear as well.

%%%%%%%%%%%%%%%%%%
\begin{lemma}
\label{lem: static: reopt: existence tour}
  Let $\overline{\tourd}$ be a uniform tour for a driver starting and ending in depot~$v_D$.
  Then there exists a set of close distance uniform transport requests for~$\overline{\tourd}$.
\end{lemma}

%%%%%%%%%%%%%%%%%%

\begin{proof}
% ``\ref{lem: static: reopt: existence tour: i}'' and ``\ref{lem: static: reopt: existence tour: ii}'':
Let $\overline{\tourd} = (\move_1, \action_1, \dotsc, \move_{\ntourd - 1}, \action_{\ntourd - 1}, \move_{\ntourd})$.
Since a tour starts and ends in the depot and there must be no car in the depot it follows that
\begin{equation}\label{eq: static: reopt: existence tour: zero sum}
  \sum_{\action = (j, v, x) \in \tourd} x = 0
\end{equation}
holds for every tour $\tourd = (\move_1, \action_1, \dotsc, \move_{\ntourd - 1}, \action_{\ntourd - 1}, \move_{\ntourd})$.
Furthermore, if $\tourd$ contains actions, it follows that the first action is a pickup action and the last action a drop action.
Otherwise, there exists a move $\move = (j, \cdot, \cdot, \cdot, \cdot, x)$ with $x < 0$ or $x > \capd$ or the driver transfers vehicles into the depot, contradicting the definition of a tour.

We construct a set of transportation requests $\TR$ by assigning the station of the first pickup action in $\overline{\tourd}$ to the station of the first drop action in $\overline{\tourd}$,
the station of the second pickup action in $\overline{\tourd}$ to the station of the second drop action in $\overline{\tourd}$, and so forth until all actions are assigned.
Since the number of cars picked up or dropped in an action in $\overline{\tourd}$ is exactly one and due to Equation~\eqref{eq: static: reopt: existence tour: zero sum},
there exists a pickup action for a drop action.
Thus, it also follows that the number of actions is even in $\overline{\tourd}$.

%%%%%

Let $r \in \TR$ be a transportation request and let $\action$ be the corresponding pickup and $\action'$ the corresponding drop action for $r$.
Using this construction, we show that there are at most $\capd$ actions between $\action$ and $\action'$.

We prove the statement by induction over the number of actions~$\ntourd$ in~$\overline{\tourd}$.
For $\ntourd = 0$ and $\ntourd = 2$ there are $0 < \capd$ actions between $\action$ and $\action'$ proving the base case.

Let us assume that the induction hypothesis holds, i.e., there are at most $\capd$ actions between $\action$ and $\action'$ for all tours with $\ntourd$ actions.

Next, we prove the inductive step.
For that, let $\overline{\tourd}$ have $\ntourd + 2$ actions.

We construct a sequence of $\ntourd$ actions and $\ntourd + 1$ moves from~$\overline{\tourd}$ by removing the first pickup and the first drop action from~$\overline{\tourd}$.
Afterwards, we show that this sequence is a tour.
Then the statement follows from the induction hypothesis.

Since every action is non-empty, $\action^1$ is the first pickup action.
Let~$\action_\ell$ be the first drop action.
Furthermore, let $\move_1 = (j, v_1, t_{v_1}, w_1, t_{w_1}, x_1)$, $\move_2 = (j, v_2, t_{v_2}, w_2, t_{w_2}, x_2)$
and $\move_\ell = (j, v_\ell, t_{v_\ell}, w_\ell, t_{w_\ell}, x_\ell)$, $\move_{\ell + 1} = (j, v_{\ell + 1}, t_{v_{\ell + 1}}, w_{\ell + 1}, t_{w_{\ell + 1}}, x_{\ell + 1})$.

First, let us consider the following new moves $\hat{\move}_{1,2} = (j, v_1, t_{v_1}, w_2, t_{w_2}, x_1)$ and 
$\hat{\move}_{\ell, \ell+1} = (j, v_\ell, t_{v_\ell}, w_{\ell + 1}, t_{w_{\ell + 1}}, x_{\ell + 1})$ and
let $\hat{\move}_i = (j, v_i, t_{v_i}, w_i, t_{w_i}, x_i - 1)$, for all $3 \leq i \leq \ell - 1$, be a move constructed from the move $\move_i = (j, v_i, t_{v_i}, w_i, t_{w_i}, x_i)$.
Finally, let $\hat{\move}^\iota = \move^\iota$ for all $\ell + 2 \leq \iota \leq \ntourd + 2$.

Then $\tourd' = (\hat{\move}_{1,2}, \action_2, \hat{\move}_3, \dotsc, \action_{\ell - 1}, \hat{\move}_{\ell, \ell + 1}, \action_{\ell+1}, \hat{\move}_{\ell+2}, \dotsc, \action_{\ntourd - 1}, \hat{\move}_{\ntourd})$
is an alternative sequence of $\ntourd + 1$ moves and $\ntourd$ actions.
We show that $\tourd'$ is indeed a tour.
It is sufficient to show for every move $\hat{\move}_i = (j, \cdot, \cdot, \cdot, \cdot, x_i)$ in~$\tourd'$ that $0 \leq x_i \leq \capd$ holds.

For that, we consider the number ${n}^+$ of consecutive pickup actions before the move~${\move}_{\ell, \ell+1}$ and 
number ${n}^-$ of consecutive drop actions directly after~${\move}_{\ell, \ell+1}$ in~$\tourd$,
as well as
the number $\hat{n}^+$ of consecutive pickup actions before the move~$\hat{\move}_{\ell, \ell+1}$ and 
number $\hat{n}^-$ of consecutive drop actions directly after~$\hat{\move}_{\ell, \ell+1}$ in~$\hat{\tourd}$.
Since $\overline{\tourd}$ is a tour, it follows from Lemma~\ref{lem: static: reopt: max consecutive actions} that $0 \leq n^+ - n^- \leq \capd$ holds.
From the construction of $\tourd'$ it follows
\[
  n^+ - n^- = (\hat{n}^+ + 1) - (\hat{n}^- + 1) = \hat{n}^+ - \hat{n}^-
\]
and, thus, we have $0 \leq \hat{n}^+ - \hat{n}^- \leq \capd$.
Furthermore, it follows that the number of cars in the convoy of driver~$j$ are equal in both sequences in and after the move~$\move_{\ell + 1}$ and $\hat{\move}_{\ell,\ell + 1}$, respectively,
i.e., for every $\ell + 2 \leq i \leq \ntourd + 2$ we have $0 \leq x_i \leq \capd$, where $x_i$ is the number of cars in the move $\hat{\move}_i$.

Then, it follows for every move $\hat{\move}_i = (j, \cdot, \cdot, \cdot, \cdot, x_i)$ in~$\tourd'$ that $0 \leq x_i \leq \capd$ holds,
and, thus, that $\tourd'$ is a tour with $\ntourd$ actions.
Therefore, the induction hypothesis can be applied to $\tourd'$ and it follows that there are at most $\capd$ actions between the corresponding pickup and drop actions of a transportation request.
Since there are at most $\capd$ consecutive pickup actions at the start of a tour (Lemma~\ref{lem: static: reopt: max consecutive actions}), it follows that there are at most $\capd$ actions between the first pickup action and the first drop action.
This proves the inductive step and the statement follows.

Since we assign a transport request to every pickup and drop action, the tour serves all transport requests in $\TR$.

Finally, it follows by construction that there does not exist a transport request $(\aloc(v), \aloc(w), 1$ so that the minimal path from $\aloc(v)$ to $\aloc(w)$ of tour arcs traverses the tour arcs connecting the depot.
\end{proof}
%%%%%%%%%%%%%%%%%%

%%%%%%%%%%%%%%%%%%
\begin{lemma}
\label{cor: static: reopt: existence tour}
  Let $\overline{\tourd} = (\move_1, \action_1, \dotsc, \move_{\ntourd - 1}, \action_{\ntourd - 1}, \move_{\ntourd})$ be a uniform tour starting and ending in depot~$v_D$,
  and let $G = (V, A, w)$ be a tour graph for $\overline{\tourd}$.
  Furthermore, let $t : A \to \NN$ be a function that returns the number of cars transfered in the corresponding moves in $\overline{\tourd}$,
%   Then there exists a set of uniform transport requests~$\TR$ so that for a transport estimate function $f^{\TR}$ it holds
  and let~$\TR$ be a set of close distance uniform transport requests for~$\overline{\tourd}$.
  Then for a transport estimate function $f^{\TR}$ for~$\TR$ it holds
  \begin{enumerate}
   \item \label{cor: static: reopt: existence tour: 1} $f^{\TR}(a) = t(a)$ for all $a \in A$, and
   \item \label{cor: static: reopt: existence tour: 2} $f^{a}(a) + f^{\TR}(a) = f^{a}(a') + f^{\TR}(a')$ for all $a, a' \in A$.
  \end{enumerate}
\end{lemma}
%%%%%%%%%%%%%%%%%%

\begin{proof}
% 
%%%%%%%%%%%%%%%%%%
% \begin{sublemma}
Firstly, we show that for a transport graph $G^t = (\Vpick \cup \Vdrop \cup \Vbal, A \cup A^t, w)$ for $\overline{\tourd}$ and $\TR$ and a transport estimate function $f^{\TR}$ for $G^t$ %it holds
  \begin{enumerate}[label=(\alph*),ref=(\alph*)]
      \item \label{lem: static: reopt: existence tour: iii: i} $f^{\TR}(a) - f^{\TR}(a') = 1$ for all $a = (v, w) \in A$, $a' = (w, u) \in A$ with $w \in \Vdrop$, and
      \item \label{lem: static: reopt: existence tour: iii: ii} $f^{\TR}(a) - f^{\TR}(a') = -1$ for all $a = (v, w) \in A$, $a' = (w, u) \in A$ with $w \in \Vpick$,
  \end{enumerate}
holds.
Within the proof, we also show that~\ref{cor: static: reopt: existence tour: 1} holds.
% Afterwards, the Statement~\ref{cor: static: reopt: existence tour: 2} follows from Claim~\ref{lem: static: reopt: f leq c plus 1} \ref{lem: static: reopt: f leq c plus 1: 2}.
Secondly, the Statement~\ref{cor: static: reopt: existence tour: 2} follows from Claim~\ref{cor: static: reopt: leq C} \ref{cor: static: reopt: leq C: 2}.

%%%%%%%%%%%%%%%%%%

% \begin{subproof}
To prove the third statement, we show that $f^{\TR}(a) = x_a$, where $\move_a = (j, \cdot, \cdot, \cdot, \cdot, x_a)$ is the move corresponding to the tour arc~$a = (v, w) \in A$.
This can be seen as follows.
Since $\TR$ is a set of close distance transport requests, there does not exist a transport request $(\aloc(v), \aloc(w), 1)$ so that the minimal path of tour arcs traverses the tour arcs connecting the depot.
Thus, we have $f^\TR(a) = f^\TR(a') = 0$ where $a = (v_D, v_1) \in A$ and $a' = (\cdot, v_D) \in A$.
Since $\overline{\tourd}$ is a tour, the node $v \in V$ corresponds to a pickup action.
Thus, for the move $\move_{a_1} = (j, \cdot, \cdot, \cdot, \cdot, x_{a_1})$ corresponding to the tour arc~$a_1 = (v_1, v_2)$, we have $x_{a_1} = 1$.
Furthermore, the tour arc $a_1$ appears once on the right hand side of Equation~\eqref{eq: static: reopt: estimate length all transport arcs}, i.e., $f^\TR(a_1) = 1 = x_{a_1}$.

If $v_2 \in \Vpick$ corresponds to a pickup action then the number of cars transfered from $v_2$ to the next station is increased by one.
Since the destination of the transport request that started in $v_1$ does not correspond to $v_2$,
the corresponding tour arc $a_2$ appears twice on the right hand side of Equation~\eqref{eq: static: reopt: estimate length all transport arcs}:
once due to the transport request $(v_1, \cdot, 1)$ and once due to $(v_2, \cdot, 1)$.

Analogously, if $v_2 \in \Vdrop$ corresponds to a drop action, the number of cars transfered from $v_2$ to the next station is decreased by one.
Since the transport arc corresponding to the transport request $(v_1, v_2, 1)$ ends in $v_2$, the number of appearances of the tour arc $a_2 = (v_2, \cdot) \in A$
on the right hand side of Equation~\eqref{eq: static: reopt: estimate length all transport arcs} is decreased as well.

In the first case we have $f^{\TR}(a_1) - f^{\TR}(a_2) = -1$ and in the second case $f^{\TR}(a_1) - f^{\TR}(a_2) = 1$.

The above arguments can be applied iteratively to all nodes in $\Vpick \cup \Vdrop$, showing that~\ref{lem: static: reopt: existence tour: iii: i} and~\ref{lem: static: reopt: existence tour: iii: ii} hold.

Since the values of the transport estimate function $f^\TR$ corresponds to the number of cars transfered in the corresponding move,
the stament follows directly from Claim~\ref{cor: static: reopt: leq C} \ref{cor: static: reopt: leq C: 2}.
\end{proof}

Finally, we prove the main theorem of this section.

%%%%%%%%%%%%%%%%%%
\begin{theorem} \label{thm: static: reopt: approximation factor: symmetric and asymmetric}
For the Static Relocation Problem $(G,\z^0,\z^T,\zd,k,L)$ with one depot,
the algorithm \REOPT\
achieves an approximation factor of $\capd + 1$ for all $\capd \in \NN$.
This approximation factor holds in the symmetric and in the asymmetric case.
\end{theorem}
%%%%%%%%%%%%%%%%%%

%%%%%%%%%%%%%%%%%%

\begin{proof}
We start by proving a special case when there is only one driver in the system.
Afterwards, we generalize this special case to the general situation when there are $k$ drivers in the system.

Let $\tourd^*$ be an optimal tour for $(G,\z^0,\z^T,\zd,1,\capd)$.
Let $\TR^p$ be a set of transportation requests induced by a minimal perfect $p$-matching, and let $\tourd^p$ an optimal tour serving all transport requests in $\TR^p$,
i.e., a tour with a minimal total tour length serving all transport requests in $\TR^p$.
Finally, let $\overline{\tourd}$ be the constructed tour from Algorithm~\ref{alg: static: reopt: construct tour}.
Then we have
\[
 \ell(\tourd^*) \leq \ell(\tourd^p) \leq \ell(\overline{\tourd})
\]
where $\ell(\tourd)$ is the total tour length of the tour $\tourd$.
Thus, we only need to show that
\begin{equation}
  \ell(\overline{\tourd}) \leq (\capd + 1) \ell(\tourd^*)
\end{equation}
holds.

Let $\TR$ be a set of close distance uniform transport requests.
Since $\TR^p$ is induced by a minimal perfect $p$-matching, it holds
\[
 \sum_{(v, w, 1) \in \TR^p} d(v, w) \leq \sum_{(v, w, 1) \in \TR} d(v, w)
\]

Let $f^a$ be a traverse counter function for $\tourd^*$ and let $f^{\TR}$ be a transport estimate function.
Let $t : A \to \NN$ be a function that returns the number of cars transfered in the corresponding moves in $\tourd^*$.
By definition, a convoy is empty at the beginning of a tour, and therefore, $t(a_0) = 0$ for $a_0 = (v_D, \cdot) \in A$.
Since we have $t(a) = f^{\TR}(a)$ for all $a \in A$ (Lemma~\ref{cor: static: reopt: existence tour}~\ref{cor: static: reopt: existence tour: 1}),
it follows $f^a(a_0) + f^{\TR}(a_0) = f^a(a_0)$.
From Lemma~\ref{lem: static: reopt: f leq c plus 1} we know that $f^a(a) \leq \capd + 1$ for all $a \in A$,
and thus, it follows from Lemma~\ref{cor: static: reopt: existence tour}~\ref{cor: static: reopt: existence tour: 2} that $f^a(a) + f^{\TR}(a) \leq \capd + 1$ holds for all $a \in A$.

With above and Equation~\eqref{eq: static: reopt: estimate length all transport arcs: fTR} we can estimate the total tour length $\ell(\overline{\tourd})$ by
\begin{align*}
 \ell(\overline{\tourd}) & = \sum_{a \in A} f^a(a) w(a) + \sum_{a^t \in A^t} w(a^t) \\
                         & \leq \sum_{a \in A} f^a(a) w(a) + \sum_{a \in A} f^{\TR}(a) w(a)  \\
                         & = \sum_{a \in A} \left(f^a(a) + f^{\TR}(a) \right) w(a) \\
                         & \leq (\capd + 1) \sum_{a \in A} w(a) \\
                         & = (\capd + 1) \cdot \ell(\tourd^*)
\end{align*}
proving the statement of the theorem if there is only one driver in the system.
% \end{proof}
%%%%%%%%%%%%%%%%%%

Next we consider the general case, i.e., there are $k \in \NN$ drivers.
Since there is only one depot, and every tour starts and ends in the depot, all tours can be ``merged'' to one tour.
For that let $\schedule = (\tourd^1, \dotsc, \tourd^k)$ be a transportation schedule with $k$ tours and with $\tourd^j = (\move^j_1, \action^j_1, \dotsc, \action^j_{\lambda_j-1}, \move^j_{\lambda_j})$.
Since by assumption~$T$ is large enough, we can construct a new tour $\tourd = \tourd^1 \tourd^2 \dotsc \tourd^k$,
where (except for the first and last tour) the moves to and from the depot are replaced by one move between
the succeeding and preceding station, respectively (see Figure~\ref{fig: static: reopt: k tour one driver} for an illustration).
Thus, all $k$ tours are performed by only one driver.
Due to the triangle inequality it follows that the length of the tour~$\tourd$ is at most the total tour length of the transportation schedule $\schedule$.
Therefore, we can generalize the statement when there are~$k$ drivers in the system.

%%%%%%%%%%%%%%%%%%
\begin{figure}[ht]
    \centering
    \includegraphics[width=0.4\textwidth]{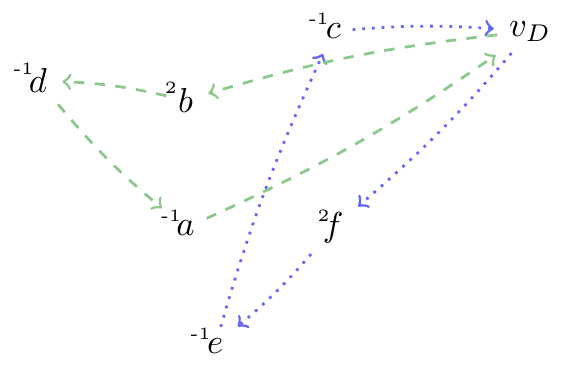}
    \includegraphics[width=0.4\textwidth]{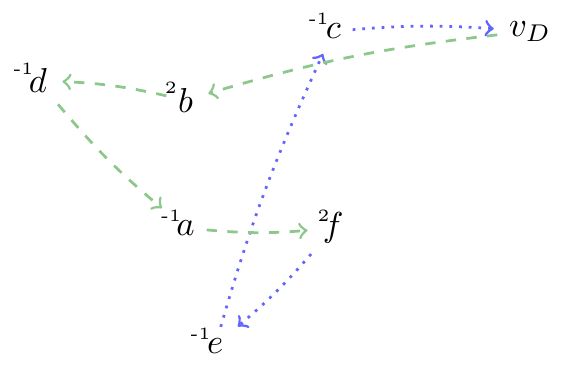}
 \caption{
  This figure illustrates how one driver can perform $k$ tours if there is only one depot.
  On the left side, the figure illustrates a transportation schedule with 2 drivers,
  On the right side, a combined transportation schedule with only one driver.
 }
 \label{fig: static: reopt: k tour one driver}
\end{figure}
%%%%%%%%%%%%%%%%%%

Finally, note that we constructed a new tour by following either the tour arcs (induced by an optimal tour) or the transport request arcs (induced by the minimal perfect $p$-matching).
Especially, no back arc is traversed.
Thus, it follows that this approximation factor is also valid for the asymmetric case when there is only one depot.
\end{proof}

%%%%%%%%%%%%%%%%%%%%%%%%%%%%%%%%%%%%
\subsection{Multiple depots}\label{sec: static: reopt: multiple depot}

%%%%%%%%%%%%%%%%%%%%%%%%%%%%%%%%%%%%
In the case that there is only one depot, the algorithm \REOPT has an approximation factor of $\capd + 1$, independently from the number of drivers,
or whether we have the symmetric or asymmetric case.
In this section, we consider the case when there are multiple depots.
Hereby, we can further distinguish between two situations: every tour must start and end in the same depot (Static Relocation Problem with multiple depots and with backhaul),
and every tour can start and end in a different depot (Static Relocation Problem with multiple depots and without backhaul).

Firstly, we consider the situation where every tour can start and end in a different depot (i.e., tours without backhaul).
We show that in the asymmetric case, the algorithm \REOPT generally does not have an approximation factor.
However, we give for a special case an approximation factor of $\capd + 1$, namely when in the optimal tour every driver returns to the depot from where its tour started.
Note that this special case makes an assumption to the properties of an optimal tour,
while the case when we have backhaul forces the tours to have certain properties in order to be a feasible solution for the Static Relocation Problem.
In the symmetric case without backhaul, the approximation factor is $2\cdot(\capd + 1)$.
We gain the factor~$2$ since the algorithm \REOPT ensures that all drivers return to their starting depot.

Secondly, we consider the situation when every tour must end in its starting depot (i.e., tours with backhaul).
In this situation, we can show that the algorithm \REOPT has an approximation factor of $\capd + 1$ in the symmetric and asymmetric situation.
We show this approximation factor, by applying nearly the same steps as we do in previous section,
where we prove the approximation factor for the case when there is only a single depot.
In fact, we mainly modify the construction of the transport graph (Definition~\ref{def: static: reopt: multiple: transport graph}) and Algorithm~\ref{alg: static: reopt: construct tour},
so that they are able to handle not only one tour but multiple tours (Algorithm~\ref{alg: static: reopt: multiple: construct tour multiple depot}).

In order to show, that there does not exist an approximation factor for the algorithm \REOPT in the asymmetric case, even when there is only one driver, we consider the following example.

\begin{example}\label{ex: static: reopt: multiple: asymmetric non possible}
We consider the asymmetric situation with two depots $v_D^1$ and $v_D^2$ and four stations $V = \{ a, b, c, d \}$.
There is only one driver in the system, the convoy capacity is~$1$.
In the stations $a$ and $c$, one car has to be picked up, and at $b$ and $d$ one car has to be dropped.
We consider the following distances
\begin{align*}
  d(v_D^1, a) & = d(b, c) = d(d, v_D^2) = \dist(v_D^1, v_D^2) = 1 \\
  d(a, b) & = d(c, d) = 2 \\
  d(c, b) & = d(a, d) = 1
\end{align*}
and for all other pairs of stations $v$ and $w$ we set
\[
 d(v, w) = x,
\]
with $x \gg 2$ arbitrary.
An optimal tour $\tourd^*$ serving all requests is induced by the path $(v_D^1, a, b, c, d, v_D^2)$.
The total tour length $\ell(\tourd^*) = 5$.

%%%%%%%%
\begin{figure}[ht]
    \centering
    \includegraphics[width=0.25\textwidth]{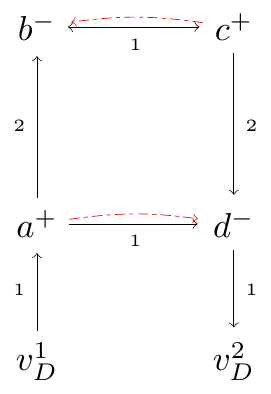}
 \caption{
  This figure shows a schedule graph for an optimal tour ${\tourd^*}$ and the set of minimal transport requests $\TR = \{ (c, b, 1), (d, a, 1) \}$.
  The arcs, with their numbers below, correspond to the distance between the connected stations.
  For all other distances, i.e., for the distances of the non-visible arcs, the distances are~$x \gg 1$.
  In the schedule graph, a pickup action at station $v \in V$ is denoted by $v^+$ and a drop action by $v^-$.
  The dash-dotted arcs correspond to transport arcs.
 }
 \label{fig: static: reopt: multiple: asymmetric 01}
\end{figure}
%%%%%%%%

A set of transport requests induced by a minimal perfect $p$-matching is given by $\TR = \{ (c, b, 1), (d, a, 1) \}$.
The total tour length $\ell(\tourd)$ for every tour, serving the transport requests in $\TR$, is at least~$x + 4$.
Since $x$ was arbitrarily selected, there cannot exist a constant $c \in \RR$ so that
\[
  x + 4 \leq \ell(\tourd) \leq c \cdot \ell(\tourd^*) = c \cdot 5
\]
holds.
\end{example}
%%%%%%%%

The reason why we cannot give an approximation factor in the previous example is that we have to traverse at least once an arc that is not traversed by an optimal tour.
Note that this example holds for all algorithms which compute tours using all arcs from a minimal perfect $p$-matching.
An open question is whether there exists a deterministic algorithm for the Static Relocation Problem with multiple depots without backhaul with a constant approximation factor.

However, in order to prove an approximation factor for the a special case of the Asymmetric Static Relocation Problem with multiple depots without backhaul,
where in the optimal transportation schedule every driver returns to its starting depot,
and in order to prove the approximation factor for the Symmetric Static Relocation Problem with multiple depots without backhaul,
we have to give some definitions and further results.

First, we define a graph similar to a transport graph, but in this case it is constructed from a transportation schedule instead of a tour.
As in the transport graph, the nodes correspond to the actions, and the arcs correspond to moves.

%%%%%%%%
\begin{definition}\label{def: static: reopt: multiple: transport graph}
From a given transportation schedule $\schedule = (\tourd^1, \dotsc, \tourd^k)$, 
where $\tourd^j = (\move_1^j, \action_1^j, \dotsc, \action_{\ntourd^j - 1}^j, \move_{\ntourd^j}^j)$
we can construct a weighted graph $G = (\Vpick \cup \Vdrop \cup \Vbal, A, w)$,
with $\Vpick = \Vpick^1 \cup \dotsm \cup \Vpick^k$, $\Vdrop^1 \cup \dotsm \cup \Vdrop^k$, $\Vbal = \Vbal^1 \cup \dotsm \cup \Vbal^k$ and $A = A^1 \cup \dots \cup A^k$,
where
\begin{enumerate}
 \item the sets of nodes $V^j = \Vpick^j \cup \Vdrop^j$ correspond to the actions in~$\tourd^j$; the set of \emph{pickup nodes} $\Vpick^j$ corresponds to the set of pickup actions,
        the set of \emph{drop nodes} $\Vdrop^j$ to the set of drop actions,
        the set of \emph{depot nodes} $\Vbal^j$ to the set containing empty actions at the depots, i.e., $\Vbal^j = \{ (j, v_D^j, 0) \mid \tourd^j \text{ starts in } v_D^j \}$;
 \item there is an arc from $v \in V^j$ to $v' \in V^j$ if $v = \action_i^j$ and $v' = \action_{i+1}^j$ for a $1 \leq i \leq \ntourd^j$, furthermore there is an arc from $(j, v_D^j, 0)$ to $\action_1^j$ and from $\action_{\ntourd - 1}^j$ to $(j, v_D^j, 0)$;
 \item the weight function $w$ corresponds to the distances between the origin and destination stations of the corresponding moves, i.e., we set $w(\action_i^j, \action_{i+1}^j) = d(\orig(\move_{i+1}^j), \dest(\move_{i+1}^j))$.
\end{enumerate}
We call such a graph a \emph{schedule graph} for $\schedule$, the set $A$ is called the set of \emph{schedule tour arcs}.

Let $\TR$ be a set of transport requests.
Then a graph $G^{tm} = (\Vpick \cup \Vdrop \cup \Vbal, A \cup A^t, w)$ is called \emph{transport schedule graph}
if $(\Vpick \cup \Vdrop \cup \Vbal, A, w)$ is a schedule graph and the set $A^t$ corresponds to the transport requests,
i.e., to every transport request $r = (\aloc(\action), \aloc(\action'), 1) \in \TR$ there is an arc $(\action, \action') \in A^t$ assigned to~$r$.
\end{definition}
%%%%%%%%

Let $\schedule = (\tourd^1, \dotsc, \tourd^k)$ be a transportation schedule, let $\TR$ be a set of transport requests, and let $G^{tm} = (\Vpick \cup \Vdrop \cup \Vbal, A \cup A^t, w)$ be the corresponding transport schedule graph.
Then we can construct a directed graph from a transport schedule graph, which has as nodes the tours and as arcs the transport requests between the tours.
Laxly said, this graph ``highlights'' the transport requests which are between two different tours.
This simplified version of a transport schedule graph enables us to prove some properties which we need in order to gain an approximation factor for \REOPT which is independent from the number of depots in the system.
Formally, we define a directed graph $G^s = (V^s, A^s)$, with
\begin{itemize}
 \item each node $v^s \in V^s$ corresponds to a tour $\tourd \in \schedule$, and
 \item every arc corresponds to a transport request $(\aloc(\action), \aloc(\action'), 1) \in \TR$, where $\action$ and $\action'$ are in different tours, i.e., $\action \in \Vpick^i \cup \Vdrop^i$, $\action' \in \Vpick^j \cup \Vdrop^j$, with $i \neq j$.
\end{itemize}
We call the graph~$G^s$ a \emph{tour connection graph} for $G^{tm}$ (see Figure~\ref{fig: static: reopt: multiple: tour connection graph} for an illustration).
Note, that a tour connection graph is loop free, but there can be multiple arcs between two nodes.

%%%%%%%%
\begin{figure}[ht]
    \centering
    \includegraphics[width=0.65\textwidth]{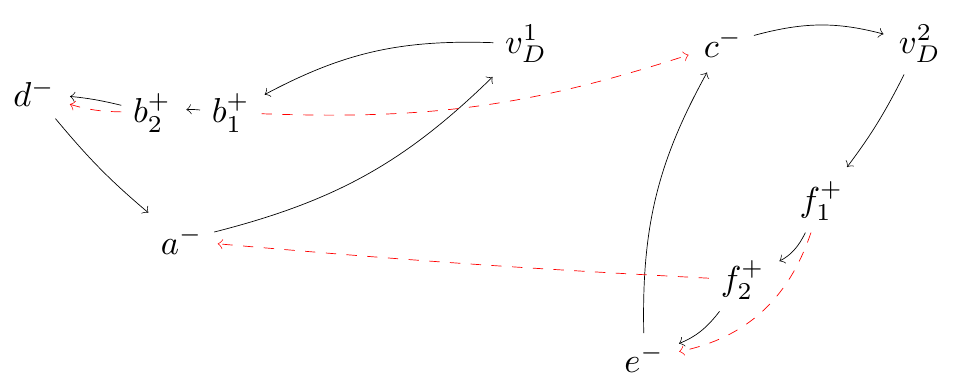}
 \caption{
  This figure shows a tour connection graph constructed from the two tours shown in Figure~\ref{fig: reopt} and from the transport requests induced by an optimal solution (see Figure~\ref{fig: flownetwork}).
  Solid arcs correspond to the schedule tour arcs, and the dashed arcs correspond to the transport requests.
  In the graph, a pickup action at station $v \in V$ is denoted by $v^+$ and a drop action by $v^-$.
  The nodes $b^+_1$ and $b^+_2$ correspond both to a pickup action at station~$b$ (and analog for the nodes $f^+_1$ and $f^+_2$).
  Furthermore, note that both tours start in the same depot, but in a tour connection graph, there is a depot for each tour.
 }
 \label{fig: static: reopt: multiple: tour connection graph}
\end{figure}
%%%%%%%%

% In the next lemma we show that every connected component of a tour connection graph is also strongly connected.
% Furthermore, we can show that if the transportation schedule has only tours with backhaul, every connected component of a transport schedule graph is also strongly connected.

%%%%%%%%%%%%%%%%%%
\begin{lemma}\label{lem: static: reopt: multiple: tour connection graph: connected thus strongly}
Let $\schedule = (\tourd^1, \dotsc, \tourd^k)$ be a transportation schedule,
let $\TR$ be a set of transport requests, and let $G^{tm} = (\Vpick \cup \Vdrop \cup \Vbal, A \cup A^t, w)$ be a corresponding transport schedule graph.
Furthermore, let $G^s$ be a tour connection graph for $G^{tm}$.
Then it is true
\begin{enumerate}
 \item\label{lem: static: reopt: multiple: tour connection graph: connected thus strongly: 1} every connected component in $G^s$ is also strongly connected,
 \item\label{lem: static: reopt: multiple: tour connection graph: connected thus strongly: 2} if every tour in $\schedule$ is a tour with backhaul, then every connected component in $G^{tm}$ is also strongly connected.
\end{enumerate}
\end{lemma}

%%%%%%%%%%%%%%%%%%

\begin{proof}
``\ref{lem: static: reopt: multiple: tour connection graph: connected thus strongly: 1}'': 
we prove that for every node~$v^s \in V^s$ and for every outgoing arc $a^s \in A^s$ of~$v^s$ there exists an incoming arc $a^s \in A^s$ of~$v^s$.
Every transport arc $a^t = (v, w) \in A^t$ connects exactly two actions.
Whenever~$v$ corresponds to an action in one tour~$\tourd$ and~$w$ corresponds to an action in another tour~$\tourd'$, there is an arc $a^s \in A^s$.
Since every action in~$\schedule$ is uniform, it follows that the number of non-empty actions in every tour is even.
Therefore, there must exist a transport arc $(v', w') \in A^t$, with $w' \in \tourd$ and $v' \in \tourd'' \neq \tourd$, i.e., there is an incoming arc in $A^s$.
Thus, we have $\abs{\delta^+(v^s)} = \abs{\delta^-(v^s)}$ for all $v^s \in V^s$.

From Euler's Theorem it now follows that in every connected component of~$G^s$ there exists an Eulerian walk.
Especially, it follows that every connected component of~$G^s$ is also strongly connected.

``\ref{lem: static: reopt: multiple: tour connection graph: connected thus strongly: 2}'': 
since $\schedule$ is a transportation schedule with backhaul, every tour of $\schedule$ is represented by a cycle of tour arcs in~$G$ and, thus, a strongly connected component.
From the two statements above it follows that every connected component in~$G^{tm}$ is also strongly connected.
\end{proof}
%%%%%%%%%%%%%%%%%%

Note, in the asymmetric case without backhaul, connected components in~$G^{tm}$ are generally not strongly connected (e.g., see Example~\ref{ex: static: reopt: multiple: asymmetric non possible}).
Furthermore, whenever we modify this graph so that every connected component becomes strongly connected, at least one arc cannot be estimated in general.
Later we see that in the symmetric case, we can modify~$G^s$ so that every connected component in~$G^s$ is also strongly connected and every added arc can be estimated.
This is achieved by connecting some actions with the ``original'' depot.
In the symmetric case, an approximation factor can then still be computed since the distance from~$v$ to~$w$ is equal to the distance from~$w$ to~$v$.

%%%%%%%%%%%%%%%%%%
\begin{algorithm}[ht]
\caption{Construct new transportation schedule (asymmetric case with multiple depots and without backhaul)}
\label{alg: static: reopt: multiple: construct tour multiple depot}
\begin{algorithmic}[1]
  \Require{a uniform transportation schedule $\schedule$ with backhaul, a set of depots~$\VD$, a set of uniform transport requests~$\TR$}
  \Ensure{a uniform transportation schedule $\overline{\schedule}$ with backhaul}
%   \State{construct transport schedule graph $G = (V_1^+ \cup \dotsm \cup V_k^+ \cup V_1^- \cup \dotsm \cup V_k^-, A_1 \cup \dotsm \cup A_k \cup A^t, w)$}
  \State{construct transport schedule graph $G = (\Vpick^1 \cup \dotsm \cup \Vpick^k \cup \Vdrop \cup \Vbal, A \cup A^t, w)$}
  \State{initialize $\overline{\tourd}^j \gets \emptyset$ for all $j \in \{ 1, \dotsc, k \}$}
  \For{$j \in \{ 1, \dotsc, k \}$}                                                                                                      \Comment{loop through every tour}
    \State{initialize $currNode \gets v_D^{j}$}                                                                                         \Comment{tour $\tourd^j$ starts in depot $v_D^j$}
    \State{initialize $currTour \gets j$}                                                                                               \Comment{the currently considered tour}
    \While{not every node in $\Vpick^{currTour}$ has been visited}
      \If{$currNode \in \Vpick^{currTour}$ \textbf{and} has not been visited}                                                              \label{alg: static: reopt: multiple: construct tour multiple depot: 8}
        \State{mark $currNode$ as visited}                                                                                              \label{alg: static: reopt: multiple: construct tour multiple depot: 9}
        \State{follow transport arc and add corresponding moves and actions to $\overline{\tourd}$}
        \State{update $currTour$ if tour is changed}                                                                                    \label{alg: static: reopt: multiple: construct tour multiple depot: 14}
      \Else{}
        \State{follow tour arc and add corresponding move to $\overline{\tourd}$}                                                       \label{alg: static: reopt: multiple: construct tour multiple depot: 16}
      \EndIf{}
    \State{update $currNode$}
    \EndWhile{}
    \State{follow current tour until arriving in depot}
  \EndFor
  \State{insert empty actions between two successive moves if necessary}
  \State{\Return $\overline{\schedule} = (\overline{\tourd}_1, \dotsc, \overline{\tourd}_k)$}
\end{algorithmic}
\end{algorithm}
%%%%%%%%%%%%%%%%%%

Algorithm~\ref{alg: static: reopt: multiple: construct tour multiple depot} constructs a new transportation schedule with backhaul from a given uniform transportation schedule with backhaul,
a set of depots and a set of uniform transport requests.
Hereby, the algorithm always prioritizes following transport request arcs over following tour arcs.
Similarly to the constructed tour from the previous section, we state and prove some properties of this transportation schedule.

It is easy to see that Algorithm~\ref{alg: static: reopt: multiple: construct tour multiple depot} constructs a new transportation schedule.
However, it is not obvious that in this transportation schedule all drivers return to their original depot, which we show in the next lemma.

%%%%%%%%
\begin{lemma}\label{lem: static: reopt: multiple: connected eq strongly}
Let $\schedule$ be a uniform transportation schedule with backhaul, let $\VD$ be a set of depots and let $\TR$ be a set of uniform transport requests.
Then Algorithm~\ref{alg: static: reopt: multiple: construct tour multiple depot} constructs a new transportation schedule with backhaul.
\end{lemma}

%%%%%%%%

\begin{proof}
In Algorithm~\ref{alg: static: reopt: multiple: construct tour multiple depot} we ``follow'' the tour arcs of a tour until we come to a pickup node that has not been visited before
(lines~\ref{alg: static: reopt: multiple: construct tour multiple depot: 8} and~\ref{alg: static: reopt: multiple: construct tour multiple depot: 16}).
If in the algorithm the current node is a non-visited pickup node, then the algorithm ``follows'' the transport
(lines~\ref{alg: static: reopt: multiple: construct tour multiple depot: 9}--\ref{alg: static: reopt: multiple: construct tour multiple depot: 14}).
From Lemma~\ref{lem: static: reopt: multiple: tour connection graph: connected thus strongly} we know that every connected component in~$G$ is also strongly connected.
Furthermore, the number of incoming and outgoing arcs for every node in $G^s$ are equal and, thus, there exists an Eulerian walk in $G^s$.
Therefore, it follows that the transportation schedule constructed in Algorithm~\ref{alg: static: reopt: multiple: construct tour multiple depot} is a transportation schedule with backhaul.
\end{proof}
%%%%%%%%

Next, we show that the approximation factor $\capd + 1$ holds even in the case of multiple depots.

%%%%%%%%%%%%%%%%%%
\begin{theorem}\label{thm: static: reopt: multiple: with backhaul: C plus 1}
For the Static Relocation Problem $(G, \VD, \z^0, \z^T, \zd, k, \capd)$ with multiple depots and with backhaul, and a complete graph $G$,
the algorithm \REOPT computes a non-preemptive transportation schedule and achieves an approximation factor of $\capd + 1$.
This approximation factor holds for the symmetric and asymmetric situation.
\end{theorem}

%%%%%%%%%%%%%%%%%%

\begin{proof}
In order to prove this statement, we show how the results of the previous section can be applied to the case when there are multiple depots.
For that we concentrate on an arbitrary tour from the optimal transportation schedule and show that we can apply the results from the previous section on this tour.
Since the tour is arbitrarily selected, the statement then follows.

Let $G^{tm}$ be a transport schedule graph for an optimal transportation schedule $\schedule$ and let $\TR^p$ be a set of transport requests computed from a minimal perfect $p$-matching.
Let $\tourd^\ell \in \schedule$ be an arbitrary tour.
We say that there is a \emph{transport request between tours}, if there exists a transport request arc $(v, w)$ in $G^{tm}$ so that $v \in \Vpick^i$ and $w \in \Vdrop^j$ with $i \neq j$.
In other words, if~$v$ corresponds to a pickup action in a tour~$\tourd^i$ and~$w$ corresponds to a drop action in tour~$\tourd^j$ with $i \neq j$.

As we did before, we construct tours~$\overline{\tourd}$ from the optimal transportation schedule and the set of transport requests.

We consider two different cases for this tour, when there are no transport requests in $\TR$ between this tour and another tour,
and when there are transport requests in $\TR$ between this tour and another tour.

%% %% %%
Case 1 (there are no transport requests in $\TR$ between this tour and another tour):
in this case, we can directly apply Theorem~\ref{thm: static: reopt: approximation factor} on $\tourd^\ell$, showing the statement for this case.

%% %% %%
Case 2 (there are transport requests in $\TR$ between this tour and another tour):
in this case, we cannot directly apply Theorem~\ref{thm: static: reopt: approximation factor} as we did in the previous case. 
However, due to Remark~\ref{rem: static: reopt: arbitrary start node}, it is not necessary to start the construction of a new tour within the depot but it can be used any arbitrary node within a transport graph.
Furthermore, all results from Section~\ref{sec: static: reopt: single depot} hold, when the construction has not been started in the depot.

From Lemma~\ref{lem: static: reopt: multiple: tour connection graph: connected thus strongly} it follows that whenever Algorithm~\ref{alg: static: reopt: multiple: construct tour multiple depot}
``leaves'' a tour at node~$v \in \Vpick^\ell$, it will eventually ``return'' to a node~$w \in \Vdrop^\ell$.
Hereby,~$w$ is the first ``entered'' node after the tour has been changed.

From the point of view of the tour~$\tourd^\ell$, the path from~$v$ to~$w$ in~$G^{tm}$ is like a transport request from~$v$ to~$w$.
Thus, we replace this path by an artificial transport request arc $(v, w)$.
By repeating this procedure for all nodes which are start or end nodes of a transport request between~$\tourd^\ell$ and another tour, we receive a transport graph~$G^\ell$.
On the constructed tour from~$G^\ell$ and the transport requests we apply Theorem~\ref{thm: static: reopt: approximation factor}.
Since the tour was arbitrarily selected, this proves the statement.
\end{proof}
%%%%%%%%%%%%%%%%%%

Finally, we consider the Symmetric Static Relocation Problem with multiple depots and without backhaul.
For a first result, we show that the algorithm \REOPT achieves an approximation factor of at most $2 \cdot (\capd + 1)$.
In order to prove this approximation factor, we construct in an intermediate step a transportation schedule with backhaul $\overline{\schedule}$ from the optimal transportation schedule $\schedule^*$.
For every tour~$\tourd^* \in \schedule^*$ we construct a new tour~$\overline{\tourd}$ by adding an arc from the last action of $\tourd^*$ back to the ``starting'' depot (if they differ in the original tour).
Since we consider the symmetric situation, this arc can be estimated by the total tour length of $\tourd^*$.
Thus, the total tour length of this constructed transportation schedule is at most twice as large as the total tour length of the optimal transportation schedule.
Applying Theorem~\ref{thm: static: reopt: multiple: with backhaul: C plus 1} on $\overline{\schedule}$ then yields:

%%%%%%%%%%%%%%%%%%
\begin{theorem}\label{thm: static: reopt: multiple: symmetric: without backhaul: 2 C plus 1}
For the Symmetric Static Relocation Problem $(G, \VD, \z^0, \z^T, \zd, k, \capd)$ with multiple depots and without backhaul, and a complete graph $G$,
the algorithm \REOPT computes a non-preemptive transportation schedule and achieves an approximation factor of $2 \cdot (\capd + 1)$.
\end{theorem}
%%%%%%%%%%%%%%%%%%

%%%%% %%%%% %%%%%
\section{Computational Results}
\label{seq: computational}

Both, the exact approach and the heuristic approach \REOPT, have been tested on randomly generated instances (with 20--80 over-/underfull stations,
50--100 stations in total, convoy capacities 5 and 10, and 10--30 drivers).
The stations are randomly distributed on a plane and the distances between two closest stations (w.r.t.~the Euclidean metric) are kept as rounded integers in the graph.
Hereby, we ensure that the graph is connected.
The time horizons are set to 100 in all test runs.
Note that the size of these instances corresponds to small car- or bikesharing systems or to clusters of larger systems, as in~\cite{SHH-2013}.

The algorithm ReOPT has been implemented in C++, and CPLEX v12.4 is used for solving ILPs.
The operating system is Linux (CentOS with kernel version 2.6.32).
The tests have been run on an Intel Xeon X5687 clocked at 3.60GHz, with 64 GB RAM.

For solving the integer linear program of the exact approach, we use Gurobi 5.6.
The test have been run on a Linux server (CentOS with kernel version 2.6.32) with 160 Intel Xeon CPUs E7-8870 clocked at 2.40GHz, with 1 TB RAM.
For the tests, we limited the number of threads to~32.
Since we could not find any feasible solution for the first instance after 40 hours, we rerun the solver on this instance with an increased time-limit of 160 hours.
After about 50 hours, a feasible solution has been found by Gurobi.
However, even after 160 hours, the optimal solution could not be found.
Due to the enormous runtime and the little to no gain, we did not rerun the solver with an increased time-limit on the remaining instances.
The solution found by the ILP solver has a total tour length of~$337$ units and the lower bound found by Gurobi is~$109$ (thus, the duality gap is approximately $67.7$\%).
Due to the enormous runtime of finding a solution with the exact approach, we did not continue with the other instances.
Furthermore, one can see, that the improvements of the solution of the ILP solver are very little compared to the extra computational time spend.

In order to compute a preemptive transportation schedule, it is not necessary to distinguish between each of the drivers.
Therefore, there are less variables within the integer linear program which models the preemptive situation than within the integer linear program which models a non-preemptive situation for the same instances.
Since the total tour lengths of transportation schedules with preemption give lower bounds for the total tour lengths of non-preemptive transportation schedules,
we state the lower bounds computed from non-preemptive transportation schedules.
The lower bounds from Table~\ref{tab: computational results} are taken from feasible preemptive transportation schedules, and the duality gaps are computed from these lower bounds.

ReOPT computes solutions within a reasonable time (in average less than 1.5 minutes for the smaller instances with at most 40 imbalanced stations,
less than 8 minutes for the middle sized instances with 60 imbalanced stations,
and about 17 minutes for the bigger instances with 80 imbalanced stations, see Table~\ref{tab: computational results}).

\begin{table}
 \centering
 \caption{
    This table shows the average computational results for several test sets of instances of the algorithm ReOPT in comparison to the found optimal value by solving the ILP (the time limit was set to 2h and 4h).
    The considered time horizon is 100 for all test instances.
    The algorithm ReOPT was run several times with different parameters for $N$ and $\Delta$.
    In this table, the following parameters and results are shown: the total amount of stations (1st column) and the number of overfull and underfull stations (2nd column).
    Hereby, the numbers in brackets are the number of overfull resp.~underfull stations.
    Furthermore, it shows the number of drivers $k$, the server capacity \capd, the average runtime in seconds of ReOPT,
    the total tour length found by ReOPT and by the ILP solver, and the average optimality gap ($(ttl_{ReOPT} - ttl_{ILP}) / {ttl_{ReOPT}}$).
 }
 \label{tab: computational results}
\begin{tabular}{c|c|c|c|c|c|c|c|c}
 stations & $\pm$stations & k  & \capd  & runtime (s) & ttl (ReOPT) & ttl (ILP 2h) & ttl (ILP 4h) & gap (\%)  \\ \hline
 50       & 20 (10/10)    & 10 & 5/10 & 2.5        & 354.25      & 330      & 330      & 7.35      \\
 50       & 40 (20/20)    & 20 & 5/10 & 72.33      & 469         & 435.75   & 432.5    & 7.74      \\
 100      & 60 (30/30)    & 30 & 5/10 & 451.5      & 665.67      & 526      & 521      & 16.73     \\
 100      & 80 (40/40)    & 30 & 5/10 & 981.92     & 658.33      & --       & --       & --
\end{tabular}
\end{table}

\section{Conclusion}
\label{seq: conclusion}

In this paper, we considered the Static Relocation Problem $(M, \z^0, \z^T, \zd, k, \capd)$,
where tours for $k$ drivers have to be computed in a (quasi) metric space $M$, where the maximal length of the tours must be smaller or equal to a given time horizon $T$.
Hereby, the drivers transfer cars between the stations by forming convoys of at most $\capd$ cars.
In order to have an exact solution we construct a time-expanded network $G_T$ from the original network $G$ and compute coupled flows (a car and a driver flow) on this network with an ILP.
Due to the coupling constraints, the constraint matrix of the network is not totally unimodular (as in the case of uncoupled flows),
reflecting that the problem is at least \NPhard.

Thus, we presented a heuristic approach to solve the Static Relocation Problem: the algorithm \REOPT.
The construction of the tours by \REOPT\ is as follows:
firstly, transport requests between ``overfull'' and ``underfull'' stations are generated by a perfect $p$-matching.
These transport requests serve as input for a Pickup and Delivery Problem, which is solved in the second step.
Finally, the tours are iteratively augmented by ``rematching'' certain origin/destination pairs, i.e., to reinsert accordingly adapted moves in such a way 
that the total tour length decreases.

The algorithm \REOPT\ has an approximation factor based on the given convoy sizes for the Symmetric Static Relocation Problem with and without backhaul
(Theorem~\ref{thm: static: reopt: approximation factor} and Theorem~\ref{thm: static: reopt: approximation factor: symmetric and asymmetric}).
In the asymmetric situation, the approach \REOPT\ can ensure a finite upper bound for the ratio of its solution to the optimal solution
only for the Asymmetric Static Relocation Problem with backhaul (Example~\ref{ex: static: reopt: multiple: asymmetric non possible} and Theorem~\ref{thm: static: reopt: multiple: symmetric: without backhaul: 2 C plus 1}).
The approximation factors for \REOPT\ in the different situations are summarized in Table~\ref{tab: conclusion: summary: approximation factors}.

\begin{table}[!htbp]
\centering
\caption{This table summarizes the approximation factors of \REOPT.
    All different situations which we consider in this paper are shown.
    In the case that there is only one depot, the tours with and without backhaul coincide (the reason why we marked two entries with a hyphon `-').}
\label{tab: conclusion: summary: approximation factors}
\begin{tabular}{c|c|c|c|c}
% \toprule
          &  \multicolumn{2}{c|}{with backhaul} & \multicolumn{2}{c}{without backhaul} \\ \cline{2-5}
          & symmetric   & asymmetric           & symmetric & asymmetric \\ \cline{2-5}
% \midrule
% \hline
single    & $\capd + 1$ & $\capd + 1$          & - & - \\ \hline
multiple  & $\capd + 1$ & $\capd + 1$          & $2 (\capd + 1)$  & $\infty$ \\
% \bottomrule
\end{tabular}
\end{table}

Both approaches have been tested on randomly generated instances (with 20--80 over-/underfull stations, 50 and 100 stations in total, convoy capacities 5 and 10, and 10--30 drivers). 
The time horizon was set to 100 in all test runs.
Note that the size of these instances corresponds to small car- or bikesharing systems or to clusters of larger systems, as in \cite{SHH-2013}.
While the optimal solution could not be found even after 160 hours, the algorithm
ReOPT computes solutions within a reasonable time (in average less than 1.5 minutes for the smaller instances with at most 40 imbalanced stations,
less than 8 minutes for the middle sized instances with 60 imbalanced stations,
and about 17 minutes for the bigger instances with 80 imbalanced stations), and a reasonable gap to the solution computed by the ILP (see Table~\ref{tab: computational results}).
In order to be able to give at least a meaningful duality gap, we compared the solution with some lower bounds for preemptive transportation schedules, which can be computed in a shorter time.

There are several practical and theoretical open questions according to the Static Relocation Problem.
Improving the runtime and solutions of ReOPT is one goal.
Usually, every driver used gives additional costs.
Thus, it is desirable to know the minimal number of drivers needed in order to solve the Static Relocation Problem within the given time horizon.
To the best of our knowledge this is still an open question.
Due to the time horizon, it is possible that there does not exist a feasible solution for a given instance at all.
Having feasibility conditions is useful in two directions: to save unnecessary computation time for an algorithm and to generate test instances which can give feasible solutions.
Thus, finding feasibility conditions as well as lower bounds for the time horizon is another goal for the future.

\bibliographystyle{plain}

%\bibliography{lit3}

\end{document}